\documentclass[12pt]{article}

\usepackage{times}
\usepackage{amssymb,latexsym}
\usepackage{graphicx,amsmath,amsfonts,amsthm}
\usepackage{comment}
\usepackage{tikz}
\usetikzlibrary{decorations.pathreplacing,calc}
\usepackage{pgfplots}
\usepackage{url}
\usepackage{mathrsfs}
\usepackage{paralist}
\usepackage{nicefrac}
\usepackage{microtype}
\usepackage{booktabs}
\usepackage{subcaption}
\usepackage{pifont}
\usepackage{makecell}
\usepackage{natbib}
\usepackage{multirow,bigdelim}
\allowdisplaybreaks
\newtheorem{theorem}{Theorem}

\newtheorem{definition}{Definition}
\newtheorem{example}{Example}
\newtheorem{proposition}{Proposition}

\newcommand{\score}{{{\mathrm{sc}}}}
\newcommand{\reals}{\mathbb R}

\newcommand{\naturals}{\mathbb N}

\newcommand{\calA}{\mathcal{A}}

\newcommand{\calR}{\mathcal{R}}

\renewcommand{\H}{{{\mathrm{h}}}}
\newcommand{\SL}{{{\mathrm{s}}}}
\newcommand{\W}{{{\mathcal{W}}}}

\newcommand{\npav}{\bar n}

\newcommand{\sainte}{Sainte-Lagu\"e}
\newcommand{\monroe}{{{{\mathrm{Monroe}}}}}

\newcommand{\opt}{{{{\mathrm{opt}}}}}

\definecolor{OKgreen}{RGB}{3, 89, 37}
\definecolor{NOred}{RGB}{143, 6, 31}

\newcommand*\cmark{\textcolor{OKgreen}{\ding{51}}}
\newcommand*\xmark{\textcolor{NOred}{$\boldsymbol{\times}$}}
\newcommand{\phragmen}{Phragm\'{e}n}

\newcommand{\todo}[1]{{\textbf{TODO:}\ \textit{{#1}}\ \textbf{:ODOT}}}

\makeatletter
\newtheorem*{rep@theorem}{\rep@title}
\newcommand{\newreptheorem}[2]{%
\newenvironment{rep#1}[1]{%
 \def\rep@title{#2 \ref{##1}}%
 \begin{rep@theorem}}%
 {\end{rep@theorem}}}
\makeatother
\newreptheorem{theorem}{Theorem}
\newreptheorem{proposition}{Proposition}
\newreptheorem{lemma}{Lemma}
\newreptheorem{corollary}{Corollary}

\def\2vec#1#2{\left(\begin{array}{c}{#1}\\{#2}\end{array}\right)}

\newcommand{\av}{{{\mathrm{av}}}}
\newcommand{\pav}{\mathrm{pav}}
\newcommand{\hamming}{\mathrm{ham}}
\newcommand{\slav}{\mathrm{slav}}
\newcommand{\phrag}{\mathrm{phrag}}
\newcommand{\spav}{\mathrm{spav}}
\newcommand{\pgeom}{\mathrm{p}\text{-}\mathrm{geom}}
\newcommand{\cc}{\mathrm{cc}}

\def\argmax{\mbox{argmax}}

\usepackage[sort&compress,nameinlink]{cleveref}
\crefname{table}{Table}{Tables}
\crefname{figure}{Figure}{Figures}
\crefname{theorem}{Theorem}{Theorems}
\crefname{definition}{Definition}{Definitions}
\crefname{corollary}{Corollary}{Corollaries}
\crefname{observation}{Observation}{Observations}
\crefname{lemma}{Lemma}{Lemmas}
\crefname{example}{Example}{Examples}
\crefname{reduction}{Reduction}{Reductions}
\crefname{construction}{Construction}{Constructions}
\crefname{subsection}{Subsection}{Subsections}
\crefname{section}{Section}{Sections}
\crefname{proposition}{Proposition}{Propositions}
\crefname{algorithm}{Algorithm}{Algorithms}
\Crefname{equation}{Inequality}{Inequalities}

\title{Utilitarian Welfare and Representation Guarantees\\of Approval-Based Multiwinner Rules}

\author{Martin Lackner\\
  TU Wien\\
  Vienna, Austria\\
  {\small\texttt{lackner@dbai.tuwien.ac.at}}
  \and 
Piotr Skowron\\
  University of Warsaw\\
  Warsaw, Poland\\
  {\small\texttt{p.skowron@mimuw.edu.pl}}
}
\date{}

\begin{document}

\maketitle

\begin{abstract}
To choose a suitable multiwinner voting rule 
is a hard and ambiguous task.
Depending on the context, it varies widely what constitutes the choice of an ``optimal'' subset of alternatives.
In this paper, we provide 
a quantitative analysis of multiwinner voting rules using methods from the theory of approximation algorithms---we 
estimate how well multiwinner rules approximate two extreme objectives: 
a representation criterion defined via the Approval Chamberlin--Courant rule and a utilitarian criterion defined via Multiwinner Approval Voting.
With both theoretical and experimental methods, we classify multiwinner rules in terms of their quantitative alignment with these two opposing objectives.
Our results provide fundamental information about the nature of multiwinner rules and, in particular, about the necessary tradeoffs when choosing such a rule.
\end{abstract}

\section{Introduction}

A multiwinner rule is a voting method for selecting a fixed-size subset of alternatives, a so-called \emph{committee}. More formally, it is a function that---given (i) a set of candidates, (ii) preferences of a population of voters over these candidates, and (iii) an integer $k$---returns a subset of exactly $k$~candidates.
Multiwinner rules are applicable to problems from and beyond the political domain, for instance for selecting a representative body such as a parliament or university senate~\citep{ccElection, FSST-trends}, shortlisting candidates (e.g., in a competition)~\citep{barberaNonControversial}, designing search engines~\citep{dwo-kum-nao-siv:c:rank-aggregation,proprank}, building recommender systems~\citep{owaWinner}, and as mechanisms for locating facilities~\citep{FFG16}.

Ideally, a multiwinner rule should select the ``best'' committee, but the 
suitability of a chosen committee strongly depends on the specific context. For instance, if voters are experts (e.g., judges in a sports competition) whose preferences reflect their estimates of the objective quality of candidates, then the goal is typically to pick $k$ individually best candidates, e.g., those candidates who receive the highest scores from judges. Intuitively, in this and similar scenarios, the quality of candidates can be assessed separately, and a suitable multiwinner rule should pick the $k$ best-rated ones. On the contrary, if the voters are citizens and the goal is to choose locations for $k$ public facilities (say, hospitals), then our goal is very different: assessing the candidates separately can result in building all the facilities in one densely populated area; yet, it is preferable to spread them to ensure that as many citizens as possible have access to \emph{some} facility in their vicinity.

These two examples illustrate two very different goals of multiwinner rules, which can be informally described by two principles \citep{FSST-trends,lac-sko:t:abc-approval-multiwinner}:
The \emph{principle of diversity} states that a rule should select a committee that represents the voters as well as possible. Here, we mean by a representative of voter in a committee $W$ the one candidate that the voter prefers most among those from $W$. Informally speaking, a committee represents the voters well, if---on average---the voters are satisfied with their respective representatives in the committee; this translates to choosing a hospital distribution that covers as many citizens as possible.
The \emph{principle of individual excellence} suggests picking those candidates that individually receive the highest total support from the voters; this translates to selecting a group of best contestants in the previous example.
In other words, according to the principle of diversity, one should ensure that all voters are well-represented in the committee, whereas according to the principle of individual excellence the goal should be rather selecting the strongest candidates (paying no or little attention to individual voters).
The focus for the former principle is on the voters; for the latter on the candidates.

Many real-life scenarios do not clearly fall into one of the two categories. For example, rankings provided by a search engine should list the most relevant websites but also provide every user at least one helpful link.
In such cases, a mechanism designer would be interested in choosing a rule that guarantees a good representation of voters as well as some degree of excellence of the chosen candidates, putting more emphasis on one of them depending on the particular context.

Consequently, to choose a multiwinner voting rule in a principled fashion for a specific application, it is essential to understand to which degree established multiwinner rules adhere to one of these two principles. We propose two quantitative criteria that correspond to the principles of diversity and of individual excellence, 
and provide a classification that clarifies the behavior of well-known rules with respect to the two aforementioned principles.
In this paper, we work in the approval-based model, i.e., we assume that voters express their preferences by providing a set containing their approved candidates.
However, our approach is applicable to other preference models as well.

\subsection{Methodology}

In this paper, we consider the following two criteria in the approval-based model:
\begin{inparaenum}[(1)]
\item The utilitarian criterion counts the total number of approvals received by the selected candidates and thus measures their total support.
\item The representation criterion counts the number of voters represented by at least one member of the chosen committee, i.e., the number of voters with at least one approved committee member. 
\end{inparaenum} 

We investigate how well certain rules perform in terms of the two criteria. We provide their \emph{utilitarian} and \emph{representation guarantees} by measuring the ratio of the utilitarian and the representation values achieved by the elected committees divided by the values achieved by committees that are optimal for the respective criterion.
Such guarantees can be viewed as \emph{quantitative properties} of voting rules. This approach is different from the traditional axiomatic approach, which is qualitative: a voting rule can either satisfy a property (axiom) or not.
Our approach provides more fine-grained information and allows us to estimate the degree to which a certain property is satisfied.
With these methods, we understand voting rules as a compromise between different (often contradictory) goals. 

According to our model, there are two rules that---by definition---behave optimally with respect to utilitarian and representation guarantees: Approval Voting (AV) and the Approval Chamberlin--Courant (CC) rule, respectively.
This opens another perspective to look at our measures 
by an analogy to the concept of worst-case approximation often used in theoretical computer science:
the utilitarian and representation guarantees of a rule $\calR$ describe how well $\calR$ approximates AV and CC, respectively.

\medskip
Building upon the notions of utilitarian and representation guarantee, we conduct two main types of analyses: %

\paragraph{Worst-case guarantees:} We derive theoretical upper bounds on how much an outcome of the considered multiwinner rules can differ from the optimal solutions according to the representation and the utilitarian criterion. We call these bounds the \emph{representation guarantee} and \emph{utilitarian guarantee}, respectively.
Our guarantees are given as functions of the committee size $k$ and return values between 0 and 1. Intuitively, a higher representation guarantee (resp. utilitarian guarantee) indicates better performance in terms of the representation (resp. utilitarian) criterion, where 1 means that the rule always returns an optimal committee according to the respective principle.
\Cref{tab:guarantees_summary} summarizes our results (the names of the rules listed in the table are explained in \Cref{sec:prelim}). Our bounds of the utilitarian and the representation guarantee are mostly asymptotically tight---they give accurate estimations when the size of the committee $k$ is large. For small values of $k$ the gaps between our lower and upper bounds become more significant.   

{\setcellgapes{3pt}
\begin{table}
\makegapedcells
\centering
\begin{tabular}{l||cc|cc|c}
 & \multicolumn{2}{c|}{utilitarian guarantee} & \multicolumn{2}{c|}{representation guarantee} & Pareto \\ 
 & lower & upper & lower & upper & \\ 
\hline 
\hline
AV & 1 & 1 & $\frac{1}{k}$ & $\frac{1}{k}$ & \cmark \\ 
CC & $\frac{1}{k}$ & $\frac{1}{k}$ & 1 & 1 & \cmark \\ 
MAV & 0 & 0 & 0 & 0 & \cmark \\ 
seq-CC & $\frac{1}{k}$ & $\frac{1}{k}$ & $1 - \nicefrac 1 e$ & $1- (1 - \nicefrac 1 k)^k$ & \xmark \\ 
PAV & $\frac{1}{2 + \sqrt{k}}$ & $\frac{2}{\lfloor \sqrt{k} \rfloor} - \frac{1}{k}$ & $\frac{1}{2}$ & $\frac{1}{2} + \frac{1}{4k-2}$ & \cmark\\ 
SLAV & $\frac{1}{2 + \sqrt{2k}}$ & $\frac{3}{2\lfloor \sqrt{k} \rfloor}$ & $\frac{3}{5}$ & $\frac{2}{3}+\frac{1}{9k-6}$ & \cmark\\ 
$p$-Geometric & $\frac{\W(k \log(p))} {k \log(p) + \W(k \log(p))}$ & $\frac{1}{k} + \frac{2\W(k \log(p))}
{k \log(p)}$ & $\frac{p-1}{p}$ & $\frac{p}{p+\frac{k}{k+2}}$ & \cmark \\ 
seq-PAV & $\frac{1}{2\sqrt{k}}$ & $\frac{2}{\lfloor \sqrt{k} \rfloor} - \frac{1}{k}$ & $\frac{1}{\log(k) + 2}$ & $\frac{1}{2} + \frac{1}{4k-2}$ & \xmark \\ 
Monroe & $\frac{1}{k}$ & $\frac{1}{k}$  & $\frac{1}{2}$ & $\frac{1}{2} + \frac{1}{k-1}$ & \xmark \\ 
Greedy Monroe & $\frac{1}{k}$ & $\frac{1}{k}$  & $\frac{1}{2}$ & $\frac{1}{2} + \frac{1}{k-1}$ & \xmark \\ 
seq-Phragm\'{e}n & $\frac{1}{5\sqrt{k} + 1}$ & $\frac{2}{\lfloor \sqrt{k} \rfloor} - \frac{1}{k}$ & $\frac{1}{2} $ & $\frac{1}{2} + \frac{1}{4k-2}$ & \xmark \\
leximin-Phragm\'{e}n & $\frac{1}{k}$ & $\frac{1}{k}$  & $\frac{1}{2} $ & $\frac{1}{2} + \frac{1}{4k-2}$ & \xmark
\end{tabular} 
\caption{Summary of worst-case guarantees for the considered multiwinner rules. The definitions of the rules listed in the table are provided in \Cref{sec:prelim}. The guarantees are functions of the committee size $k$. A higher value means a better guarantee, with~$1$ denoting optimal performance. In most cases we could only find (accurate) estimates instead of the exact values of the guarantees: the ``lower'' and ``upper'' values in the table denote that the respective guarantee is between these two bounds.  The guarantees for the $p$-Geometric rule are visualized in \Cref{fig:pgeom} (page~\pageref{fig:pgeom}). The column ``Pareto'' indicates whether the rule satisfies the Pareto efficiency axiom as discussed in \Cref{sec:efficiency_axiom}.}\label{tab:guarantees_summary}
\end{table}}

\paragraph{Average-case performance:} We complement the worst-case analysis with an experimental study yielding average ratios for several datasets. In extensive experiments based on real-world data and numerical simulations, we estimate how on average the outcomes of the considered rules differ from the optimal committees with respect to the utilitarian and the representation criterion.
\bigskip

Let us now explain in  more detail our choice of criteria, i.e., the choice to use AV and CC as cornerstones of our analysis, and the underlying assumptions we make.
First,  AV embodies the ideas of utilitarian maximization (as in the classic works on collective utility functions, see, e.g., the book by \citet{moulinAxioms}); CC 
is the most natural (and simplest) rule in the approval-based model that maximizes representation.
Second, AV and CC can be axiomatically characterized by properties that capture the principle of individual excellence and the principle of diversity~\citep{lac-sko:t:abc-approval-multiwinner}.
Third, both AV and CC are well-known, well-understood rules.
Fourth, due to their orthogonal nature they are good choices for a two-dimensional evaluation.

The rationale for using AV and CC is based on an important underlying assumption:
we assume that the amount of utility a voter derives from a committee $W$ is linear in the number of candidates in $W$ that the voter approves. Taking this utility-based viewpoint, the utilitarian guarantee of a rule~$\calR$ measures the performance of~$\calR$ with respect to utilitarian social welfare (i.e., the sum of utilities). 
Furthermore, the representation guarantee measures a voting rule's effort to guarantee as many voters as possible non-zero utility.
Thus, a voting rule with a high representation guarantee ensures that the preferences of as many voters as possible are taken---at least minimally---into account.
In this light, the representation guarantee can be viewed as an egalitarian criterion; we discuss this and other egalitarian criteria in more detail in Section~\ref{sec:concl}.

\subsection{Contribution}

The most important findings of our analyses can be summarized as follows:
\begin{itemize}
\item Among the studied voting rules, we identify six rules that provide a particularly appealing compromise between the utilitarian and representation criteria: Proportional Approval Voting (PAV), its sequential variants seq-PAV and rev-seq-PAV, \phragmen's sequential rule, the $2$-Geometric rule, and \sainte\ Approval Voting (SLAV).
These rules exceed in both the worst-case and average-case analysis.

\item The latter of these rules, \sainte\ Approval Voting (SLAV), is a new approval-based multiwinner rule. It is similar to PAV (in that it is a Thiele method, cf.~\citealt{lac-sko:t:abc-approval-multiwinner}) but is based on the \sainte\ apportionment method instead of D'Hondt apportionment (as PAV does). This modification strengthens the representation guarantee of SLAV compared to PAV, while slightly reducing its utilitarian efficiency as observed in the experiments.
Nonetheless, its (worst-case) utilitarian guarantee is of the same (asymptotic) order as that of PAV, and it is significantly higher than that of CC.

\item Subject to utilitarian and the representation guarantees, our results reveal major deficits of prominent multiwinner rules: Monroe, its sequential counterpart, \phragmen's leximin rule (leximin-\phragmen), and Minimax Approval Voting (MAV). In particular, the contrast between PAV and Monroe is noteworthy, as both are known as proportional rules \citep{justifiedRepresenattion,pjr17}. 
We note that one should judge multiwinner rules from different angles and our criteria provide only one specific viewpoint; however, if one considers the utilitarian and the representation guarantees as important objectives, then the loss of the utilitarian and the representation guarantee of the aforementioned rules is striking and speaks against the use of these rules.

\item We observe that rules that are more expensive to compute---PAV, Monroe, CC, and leximin-\phragmen\ are NP-hard \citep{owaWinner,azi-gas-gud-mac-mat-wal:c:multiwinner-approval,complexityProportionalRepr,aaai/BrillFJL17-phragmen}---do not necessarily achieve a good compromise between AV and CC: while PAV achieves an excellent compromise and CC is by definition optimal with respect to one criterion, this is contrasted by Monroe and leximin-\phragmen, which are computationally expensive without a clear benefit in this regard.

\item Also the 2-Geometric rule achieves a very good compromise between the two criteria.
More generally, we show that the $p$-Geometric rule spans the whole spectrum from AV to CC, controlled through the parameter $p$. Hence, by adjusting the parameter $p$, one can obtain any desired compromise between utilitarian and representation guarantees. 

\item Finally, we show that while proportional rules tend to achieve a good compromise between utilitarian and representation objectives, proportionality does not yield an optimal compromise: we find a non-proportional rule that provides better utilitarian and representation guarantees than \emph{arbitrary} proportional rules. Thus, proportionality should be viewed as a third standalone criterion rather than as a compromise between utilitarian and representation objectives.

\end{itemize}

At the end of the paper (\Cref{sec:efficiency_axiom}), we complement our results with an analysis of an efficiency axiom, which can be viewed as an incarnation of Pareto efficiency in the context of multiwinner elections. We say that a committee $W_1$ dominates a committee $W_2$ if each voter approves as many members of $W_1$ as of $W_2$ and some voter approves strictly more members of $W_1$ than of $W_2$. Pareto efficiency states that a rule should never select a dominated committee; thus Pareto efficiency could be viewed as a basic requirement for any utilitarian rule. Since Pareto efficiency appears to be very fundamental, it may come as a surprise that many known rules do not satisfy this property (in particular, Monroe's rule, the Phragm\'{e}n rules, and all sequential rules; see \Cref{tab:guarantees_summary}). Further, we observe that Pareto efficiency is unsuitable to distinguish rules with strong and weak utilitarian guarantees. We view this discrepancy as yet another argument why analyzing utilitarian and representation guarantees of multiwinner rules is important for better understanding their nature, and why a quantitative approach is required.

\subsection{Related Work}

Our work is based on the idea of approximation algorithms, where computationally hard problems are solved by polynomial-time algorithms that can guarantee a certain (imperfect) solution quality. In our paper, we study how well popular multiwinner rules can approximate other, archetypical rules. 
\citet{pro-ros:c:distortion} evaluate the quality of social choice decisions (and so, the quality of voting rules) by analyzing how well voting rules based on ordinal preferences can approximate optimal decisions based on cardinal utilities.
The same question for randomized social choice functions has been explored by~\citet{BCHLPS:randomizedDistortion}. In a related facility-location model the idea of approximation is used to reason about the quality of strategy-proof mechanisms~\citep{ProTen:approxDesignNoMoney}. 
The work of \citet{branzei2013bad} on the dynamic price of anarchy can also be viewed from such a perspective: to which degree can the outcome of voting rules based on sincere preferences be approximated by the same voting rules with insincere preference (obtained via ``selfish'' best-response dynamics)?
In a similar vein, \citet{oren2014online} study the performance of online social choice procedures in comparison to optimal (offline) procedures.
\citet{anshelevich2018approximating} approximate an optimal social choice in a metric model with voting rules using rankings as input, i.e., using limited information.

The normative study of multiwinner election rules typically focuses on axiomatic analysis. Recent work on the axiomatic analysis of approval-based rules has focused on particularly on axioms describing proportionality
\citep{justifiedRepresenattion, pjr17, aaai/BrillFJL17-phragmen, proprank, AEHLSS18, SFF16a, lac-sko:t:abc-approval-multiwinner}, strategyproofness \citep{pet:prop-sp,lac-sko:t:multiwinner-strategyproofness,KluivingEtAlECAI2020} and monotonicity \citep{sanchez2019monotonicity}. Similar axiomatic properties for the ordinal model have been discussed by \citet{dum:b:voting}, \citet{monroeElection}, \citet{ccElection}, \citet{elk-fal-sko-sli:c:multiwinner-rules}, and \citet{AEFLS17:multiwinner-condorcet}; and for the model with weak preferences by \citet{BaumeisterBRSS16} and \citet{aziz2020expanding}. For an overview of multiwinner rules in general, with the focus on the ideas of individual excellence, diversity, and proportionality, we refer the reader to a survey by \citet{FSST-trends}.

Another approach to understanding the nature of different multiwinner rules is to analyze how these rules behave on certain subdomains of preferences, where their behavior is much easier to interpret, e.g., on two-dimensional geometric preferences~\citep{2dMultiwinner}, on party-list profiles~\citep{BLS17a}, or on single-peaked and single-crossing domains~\citep{AEFLS17:multiwinner-condorcet}. Other approaches include analyzing certain aspects of multiwinner rules in specifically designed probabilistic models \citep{laslier2016StrategicVoting, RePEc:ucp:jpolec:doi:10.1086/670380, skow:multiwinner-models, ProcacciaRS12}, quantifying regret and distortion in utilitarian models~\citep{CNPS16a}, assessing their robustness~\citep{sagt/BredereckFKNST17}, and evaluating them based on data collected from surveys~\citep{RFZ88,LasStr17}.

\subsection{Structure of the Paper}

In \Cref{sec:prelim} we provide basic definitions and introduce the multiwinner rules relevant for our study.
\Cref{sec:worst_case_guarantees} constitutes the core of our  worst-case analysis, where we establish guarantees for the quality of committees selected by multiwinner rules.
The reader could wonder whether proportionality can be viewed as an optimal compromise between utilitarian and representation objectives. We disprove this thesis in \Cref{sec:propcompromise} and show that proportional rules cannot provide as good utilitarian and representation guarantees as rules specifically designed to provide a compromise between these two objectives.
The average-case counterpart of our analysis, based on numerical simulations and real-world datasets, is presented in \Cref{sec:average_guarantees}.
In \Cref{sec:efficiency_axiom}, we discuss Pareto efficiency as an example to highlight the difficulty of performing an analysis similar to ours in a strict axiomatic framework.
We conclude the paper with directions for future research in \Cref{sec:concl}.
Proof details and additional details concerning our experimental analysis have been delegated to the appendix. 

\medskip

This work is based on the short paper titled ``\emph{A quantitative analysis of multi-winner rules}'' that appeared in the proceedings of the 28th International Joint Conference on Artificial Intelligence (IJCAI 2019) \citep{ijcai/LS-quantitative}.

\section{Basic Definitions and Multiwinner Voting Rules}\label{sec:prelim}

For each $t \in \naturals$, we let $[t] = \{1, \ldots, t\}$. For a set $X$, we write $\mathcal{S}(X)$ to denote the powerset of $X$, i.e., the set of all subsets of $X$. By $\mathcal{S}_k(X)$ we denote the set of all $k$-element subsets of $X$.

Let $C = \{c_1, \ldots, c_m\}$ and $N = \{1, \ldots, n\}$ be  sets of $m$ \emph{candidates} and $n$ \emph{voters}, respectively.
Voters reveal their preferences by indicating which candidates they like: by $A(i) \subseteq C$ we denote the \emph{approval set} of voter~$i$ (that is, the set of candidates that~$i$ approves of).
For a candidate $c \in C$, by $N(c) \subseteq N$ we denote the set of voters who approve~$c$.
Given a set of candidates $X\subseteq C$, we write $N(X)$ to denote the set of voters who approve at least one candidate in $X$, that is $N(X) = \{i\in N: X\cap A(i)\neq \emptyset\}$.
We call the collection of approval sets $A=(A(1), A(2), \dots, A(n))$, one per each voter, an \emph{approval profile}. 
We use the symbol $\mathcal{A}$ to represent the set of all possible approval profiles.

We call the elements of $\mathcal{S}_k(C)$ size-$k$ \emph{committees}. Throughout the paper, we use the symbol $k$ to represent the desired size of the committee to be elected. An \emph{approval-based committee rule} (in short, an ABC rule) is a function $\calR \colon \mathcal{A} \times \naturals \to \mathcal{S}(\mathcal{S}_k(C))$ that takes as an input an approval profile and an integer $k \in \naturals$ (the required committee size), and returns a set of size-$k$ committees.\footnote{Rules which for some profiles return multiple committees as tied winners are often called irresolute. In practice, one usually uses some tie-breaking mechanism to single out a winning committee.} For a given committee $W$ and a voter $i$, we will often refer to the candidates in $W \cap A(i)$ as \emph{representatives} of $i$ in $W$ (sometimes we will omit $W$ when it will be clear from the context). Below, we recall the definitions of ABC rules which are the objects of our study. 

\begin{description}
\item[Multiwinner Approval Voting (AV).] This rule selects $k$ candidates which are approved by most voters. More formally, for a profile $A$ the AV-score of committee $W$ is defined as $\score_{\av}(A, W) = \sum_{c \in W}|N(c)|$, and AV selects committees $W$ that maximize $\score_{\av}(A, W)$.  
\item[Approval Chamberlin--Courant (CC).] For a profile $A$ we define the CC-score of a committee~$W$ as $\score_{\cc}(A, W)  = \sum_{i \in N}\min\big(1, |W \cap A(i)|\big)= |N(W)|$; CC outputs committees $W$ that maximize $\score_{\cc}(A, W)$. In words, CC aims at finding a committee $W$ such that as many voters as possible have their representatives in $W$. The CC rule was first mentioned by Thiele~\citep{Thie95a}, and then introduced in a more general context by Chamberlin and Courant~\citep{ccElection}.     
\item[Proportional Approval Voting (PAV).] This rule~\citep{Thie95a} selects committees with the highest PAV-scores, defined as $\score_{\pav}(A, W) = \sum_{i \in N} \H\left(|W \cap A(i)|\right)$, where $\H(t)$ is the $t$-th harmonic number, i.e., $\H(t) = \sum_{i=1}^{t} \nicefrac{1}{i}$. By using the harmonic function $\H(\cdot)$, voters who already have more representatives in the committee get less voting power than those with fewer representatives. While using other concave functions instead of $\H(\cdot)$ would give similar effects, the harmonic function is particularly well justified---it implies a number of appealing properties of the rule~\citep{justifiedRepresenattion,AEHLSS18}, and it allows one to view PAV as an extension of the D'Hondt method~\citep{BLS17a,lac-sko:t:abc-approval-multiwinner}.
\item[$\boldsymbol{p}$-Geometric.] This rule~\citep{owaWinner} is defined analogously to PAV but uses an exponentially decreasing function instead of $\H(\cdot)$. Formally, for a given parameter $p\geq1$ the $p$-Geometric rule assigns to each committee $W$ the score $\score_{p\text{-}\mathrm{geom}}(A, W) = \sum_{i \in N} \sum_{j = 1}^{|W \cap A(i)|} \frac{1}{p^{j}}$, and picks the committees with the highest scores. Note that the $1$-Geometric rule is simply AV. 

\item[Sainte-Lagu\"e Approval Voting (SLAV).] This is a rule that---to the best of our knowledge---has not been mentioned in the literature so far. It is a straight-forward generalization from the Sainte-Lagu\"e (or Webster) apportionment method~\citep{BaYo82a} to the approval-based multiwinner setting in the same way that PAV can be viewed as an extension of the D'Hondt apportionment method \citep{BLS17a,lac-sko:t:abc-approval-multiwinner}.
The Sainte-Lagu\"e apportionment method distinguishes itself by fairly treating smaller and large parties, in contrast to D'Hondt which decides in favor of larger ones.
Sainte-Lagu\"e Approval Voting is defined analogously to PAV:
SLAV selects committees with the highest SLAV-scores, defined as $\score_{\slav}(A, W) = \sum_{i \in N} \SL\left(|W \cap A(i)|\right)$, where $\SL(t) = \sum_{i=1}^{t} \nicefrac{1}{2i-1}$.

\end{description}

We remark that all of the aforementioned rules belong to the class of Thiele methods~\citep{FSST-trends,lac-sko:t:abc-approval-multiwinner}. These are parameterized by a (non-decreasing) function $f:\naturals\to\naturals$. The score of a committee $W$ subject to a profile $A$ is defined as $\score_{\pav}(A, W) = \sum_{i \in N} f\left(|W \cap A(i)|\right)$; the $f$-Thiele method returns committees with maximum score.
The following rules do not belong to this class.

\begin{description}

\item[Sequential CC/PAV/SLAV/$\boldsymbol{p}$-Geometric.] For each rule $\calR \in \{\mathrm{CC}, \mathrm{AV}, \mathrm{PAV}, \mathrm{SLAV},$ $p\text{-}\mathrm{Geometric}\}$, we define its sequential variant, denoted as seq-$\calR$, as follows. We start with an empty solution $W = \emptyset$ and in each of the $k$ consecutive steps we add to $W$ a candidate $c$ that maximizes $\score_{\calR}(A, W \cup \{c\})$, i.e., the candidate that improves the committee's score most. We break ties with an arbitrary but fixed order among candidates.

\item[Reverse Sequential CC/PAV/SLAV/$\boldsymbol{p}$-Geometric.] For each rule $\calR \in \{\mathrm{CC}, \mathrm{AV}, \mathrm{PAV}, \mathrm{SLAV},$ $p\text{-}\mathrm{Geometric}\}$, we define Reverse Sequential $\calR$, abbreviated as rev-seq-$\calR$, as follows. Given an approval profile $A$, the rule starts with a full set of candidates $W = C$, and in each step it removes from $W$ the candidate $c$ that maximizes $\score_{\calR}(A, W \setminus \{c\})$. The rule stops when there are $k$ candidates left in $W$. As previously, we break ties with a fixed order. In this paper, we only consider Reverse Sequential PAV.

\item[Monroe.] Monroe's rule~\citep{monroeElection}, similarly to CC, aims at maximizing the number of voters who are represented in the elected committee. The difference is that for calculating the score of a committee, Monroe additionally imposes that each candidate should be responsible for representing roughly the same number of voters. 
Formally, a Monroe assignment of the voters to a committee $W$ is a function $\phi \colon N \to W$ such that each candidate $c \in W$ is assigned roughly the same number of voters,
i.e., that $\lfloor \nicefrac{n}{k} \rfloor \leq |\phi^{-1}(c)| \leq \lceil \nicefrac{n}{k} \rceil$. Let $\Phi(W)$ be the set of all possible Monroe assignments to $W$.
The Monroe-score of $W$ is defined as
$\score_{\monroe}(A, W) = \max_{\phi \in \Phi(W)} \sum_{i \in N} |\{\phi(i)\} \cap A(i)|$;
the rule returns $\argmax_{W} \score_{\monroe}(A, W)$.
\item[Greedy Monroe~\citep{sko-fal-sli:j:multiwinner,elk-fal-sko-sli:c:multiwinner-rules}.] This is a sequential variant of Monroe's rule. It proceeds in $k$ steps: In each step it selects a candidate $c$ and a group $G$ of $\lfloor \nicefrac{n}{k} \rfloor$ or $\lceil \nicefrac{n}{k} \rceil$ not-yet removed voters\footnote{To be precise, for $n=k\cdot \lfloor \nicefrac{n}{k} \rfloor+c$, the first $c$ groups of voters to be removed have size $\lceil \nicefrac{n}{k} \rceil$ and the remaining $k-c$ have size $\lfloor \nicefrac{n}{k} \rfloor$.} so that $|N(c) \cap G|$ is maximal; next, candidate $c$ is added to the winning committee and the voters from $G$ are removed from further consideration. 
\item[\phragmen's Sequential Rule (seq-\phragmen).] Perhaps the easiest way to define the family of \phragmen's rules~\citep{Phra94a,Phra95a,Phra96a,Janson16arxiv,aaai/BrillFJL17-phragmen} is by describing them as load distribution procedures. We assume that each selected committee member $c$ is associated with one unit of load that needs to be distributed among those voters who approve~$c$ (though it does not have to be distributed equally). Seq-\phragmen{} proceeds in $k$ steps. In each step it selects one candidate and distributes its load as follows: let $\ell_j(i-1)$ denote the total load assigned to voter~$j$ just before the $i$-th step ($\ell_j(-1) = 0$ for each $j$). In the $i$-th step the rule selects a candidate $c$ and finds a load distribution $\{x_j\colon j\in N\}$ that satisfies the following three conditions:
\begin{inparaenum}[(1)]
\item $x_j > 0$ implies that voter $j$ approves $c$,
\item $\sum_{j \in N} x_j = 1$, and
\item the maximum load assigned to a voter, $\max_{j\in N} (\ell_j(i-1) + x_j)$, is minimized.
\end{inparaenum}
The new total load assigned to a voter $j \in N$ after the $i$-th step is $\ell_j(i) = \ell_j(i-1) + x_j$. 
\item[\phragmen's Leximin Rule (leximin-\phragmen)~\citep{aaai/BrillFJL17-phragmen}.] This is a variant of \phragmen's rules where committee members and their associated load distributions are chosen simultaneously in a single step, i.e., by solving an optimization problem. Similarly, as in the case of its sequential counterpart, the goal is to find a committee and an associated load distribution that minimizes the load of the voter with the highest load.
\phragmen's leximin rule uses a lexicographic tie-breaking if two committees have the same highest load, then the second highest load is compared, etc. We refer the reader to the description of \citet{aaai/BrillFJL17-phragmen} for details.
\item[Minimax Approval Voting (MAV) \citep{minimaxProcedure}.] Given two subsets, $X, Y \subseteq C$, we define the Hamming distance between $X$ and $Y$ as the size of their symmetric difference:
$d_{\hamming}(X, Y) = |X \setminus Y| + |Y \setminus X|$. MAV selects committees~$W$ that minimize the largest Hamming distance among all voters, i.e., MAV minimizes $\max_{i \in N}d_{\hamming}(A(i), W)$.
\end{description}

\section{Worst-Case Guarantees of Multiwinner Rules}\label{sec:worst_case_guarantees}

Intuitively, the utilitarian objective cares about selecting candidates who receive the highest total support from the population of voters, and the representation one cares mostly about representing the minorities in the elected committee. The two objectives specify two important, but rather opposite criteria in the design of multiwinner rules. Indeed, the Chamberlin--Courant rule and Approval Voting---the rules that return optimal committees with respect to the representation and the utilitarian objectives, respectively---can be viewed as two extreme points in the spectrum of multiwinner rules~\citep{BLS17a, 2dMultiwinner, FSST-trends, lac-sko:t:abc-approval-multiwinner}. We include a simple example which illustrates the difference between AV and CC, and so between the two opposite criteria.

\begin{example}
Consider a profile where 30 voters approve candidates $\{c_1, c_2, c_3\}$, 20 voters approve $\{c_4, c_5, c_6\}$, and 5 voters approve $\{c_7, c_8, c_9\}$. Let $k = 3$. For this profile AV selects candidates $W_{\av}=\{c_1, c_2, c_3\}$, while CC selects the committee $W_{\cc}=\{c_1, c_4, c_7\}$ (among others). \label{ex1}
These two committees differ significantly: while $W_{\av}$ has an (optimal) AV-score of 90, committee $W_\cc$ has only an AV-score of 55. In contrast, $W_{\cc}$ has an (optimal) CC-score of 55, while $W_\av$ has only a CC-score of 30.
An compromise between these committees would be, e.g., $\{c_1,c_2,c_4\}$. This committee has an AV-score of 80 and a CC-score of 50, both values near optimal.
Note that this compromise committee can be viewed as proportional; we return to the relationship between proportionality, AV-, and CC-scores in \Cref{sec:propcompromise}.
\end{example}

In this section, we analyze the multiwinner rules from \Cref{sec:prelim} with respect to how well they perform in terms of the utilitarian and the representation objectives. In our study we use the established idea of approximation from computer science: we estimate how well a given rule $\calR$ approximates each of the two objectives. This differs from the typical use of the idea of approximation in the following aspects:
\begin{inparaenum}[(1)]
\item We do not seek new algorithms approximating a given objective function as well as possible, but rather analyze how well the existing known rules approximate given functions (even if it is apparent that better and simpler approximation algorithms exist, these algorithms might not share other important properties of the considered rules).
\item We are not approximating computationally hard objectives with rules that are easier to compute. On the contrary, we will be even investigating how computationally hard rules (such as PAV, Monroe, etc.) approximate the algorithmically trivial AV rule. 
\end{inparaenum}

\begin{definition}
Recall that for a profile $A$, $\score_{\av}(A, W)$ and $\score_{\cc}(A, W)$ denote the AV-score and CC-score of committee $W$, respectively (i.e., the total number of approvals the committee garners from the voters, and the total number of voters with representatives in the committee). 
The \emph{utilitarian guarantee} of an ABC rule $\calR$ is a function $\kappa_{\av} \colon \naturals \to [0,1]$ that takes as input an integer $k$, representing the size of the committee, and is defined as: 
\begin{align*}
\kappa_{\av}(k) = \inf_{A \in \mathcal{A}} \frac{\min_{W\in \calR(A, k)}\score_{\av}(A,W)}{\max_{W \in \mathcal{S}_k(C)} \score_{\av}(A, W)} \text{.}
\end{align*}
Analogously, the \emph{representation guarantee} of $\calR$ is defined by
\begin{align*}
\kappa_{\cc}(k) = \inf_{A \in \mathcal{A}} \frac{\min_{W\in \calR(A, k)}\score_{\cc}(A,W)}{\max_{W \in \mathcal{S}_k(C)} \score_{\cc}(A, W)} \text{.}
\end{align*}
\end{definition}

The utilitarian and representation guarantees can be viewed as quantitative properties of multiwinner rules. In comparison with the traditional qualitative approach (analyzing properties which can be either satisfied or not), a quantitative analysis provides much more fine-grained information regarding the behavior of a rule with respect to some normative criterion.  
In the remaining part of this section we evaluate the previously defined rules against their utilitarian and representation guarantees. 

Clearly, the utilitarian guarantee of Approval Voting and the representation guarantee of the Chamberlin--Courant rule are the constant-one function. We start by establishing the utilitarian guarantee of CC and vice versa.
As mentioned before, we focus on the presentation of our results in the main text and thus defer most proofs to the appendix.
A few proofs are presented here to give a flavor of the involved arguments and constructions.

\begin{proposition}\label{prop:AvGuaranteeOfCC}
The representation guarantee of AV is $\nicefrac{1}{k}$.
\end{proposition}

\begin{proof}
Consider an approval profile $A$, and let $W_{\av}$ be an AV-winning committee for $A$. We know that $W_{\av}$ contains a candidate who is approved by most voters---let us call such a candidate~$c_{\max}$. Clearly, it holds that $\score_{\cc}(A, W_{\av}) \geq |N(c_{\max})|$. Further, for any size-$k$ committee $W \subseteq C$ we have that $\score_{\cc}(A, W) \leq k|N(c_{\max})|$, which proves that the utilitarian guarantee of CC is at least $\nicefrac{1}{k}$. 

To see that the representation guarantee cannot be higher than $\nicefrac{1}{k}$, consider a family of profiles where the set of voters can be divided into $k$ disjoint groups: $N_1, N_2, \ldots, N_k$, with $|N_1| = x+1$ and $|N_i| = x$ for $i \geq 2$, for some large value $x$. Assume that $m = k^2$ and that all voters from $N_i$ approve candidates $c_{(i-1)k + 1}, c_{(i-1)k + 2}, \ldots, c_{ik}$. For this profile AV selects committee $\{c_1, \ldots, c_k\}$ with a CC-score of $x+1$. An optimal CC committee is, e.g., $\{c_1, c_{k+1}, \ldots, c_{k(k-1) + 1}\}$, with a CC-score of $kx + 1$.
\end{proof}

\begin{proposition}\label{prop:CCGuaranteeOfAV}
The utilitarian guarantee of CC and seq-CC is $\nicefrac{1}{k}$.
\end{proposition}

\begin{proof}
For an approval profile $A$ let $W_{\cc}$ and $W_{\av}$ be committees winning according to CC and AV, respectively. We will first prove that $\score_{\av}(A, W_{\cc}) \geq \frac{1}{k}\cdot \score_{\av}(A, W_{\av})$. If it was not the case, then by the pigeonhole principle, there would exist a candidate $c \in W_{\av}$ such that $\score_{\av}(A, W_{\cc}) < \score_{\av}(A, \{c\})$. However, this means that a committee that consists of $c$ and any $k-1$ candidates have higher CC-score than $W_{\cc}$, a contradiction. Thus, the utilitarian guarantee of CC is at least $\nicefrac{1}{k}$. 
For seq-CC, the same argument by contradiction applies as this candidate $c$ would have been chosen in the first round.

To see that this guarantee cannot be higher than $\nicefrac{1}{k}$ consider the following profile: assume there are $x$ voters ($x$ is a large integer) who approve candidates $c_1, \ldots, c_k$. Further, for each candidate $c_{k+1}, \ldots, c_{2k}$ there is a single voter who approves only her. The CC-winning committee is $\{c_1, c_{k+1} \ldots, c_{2k-1}\}$ with the AV-score of $x + k - 1$. However, the AV-score of committee $\{c_1, \ldots, c_k\}$ is $xk$, and for large enough $x$ the ratio $\frac{x + k -1}{xk}$ can be made arbitrarily close to $\nicefrac{1}{k}$. 
\end{proof}

\Cref{prop:AvGuaranteeOfCC,prop:CCGuaranteeOfAV} give a baseline for our further analysis. In particular, we would expect that ``good'' rules implementing a tradeoff between utilitarian and representation objectives should have utilitarian and representation guarantees better than $\nicefrac{1}{k}$.

We note that the representation guarantee of seq-CC is $1- (1 - \nicefrac 1 k)^k \geq 1 - \nicefrac{1}{e} \approx 0.63$. This follows from the fact that seq-CC is a \mbox{$(1- (1 - \nicefrac 1 k)^k)$}-approximation algorithm for CC~\citep{budgetSocialChoice}. (The result from \citeauthor{budgetSocialChoice} follows from a more general approximation result for submodular set functions by \citealp{submodular}.)

\iffalse
\subsection{Guarantees for Monroe and Greedy Monroe}\
\fi

Let us turn our attention to Monroe's rule, which is often considered a proportional rule. For example, Monroe's rule satisfies proportional justified representation \citep{pjr17}. Hence, one could expect that this rule offers a good compromise between AV and CC. Perhaps surprisingly, this is not the case: Monroe does not offer a better utilitarian guarantee than CC.
The same bounds hold for the Greedy Monroe rule.

\newcommand{\thmmonroecombined}{The utilitarian guarantee of Monroe and Greedy Monroe is $\nicefrac{1}{k}$, their representation guarantee is between $\frac{1}{2}$ and $\frac{1}{2} + \frac{1}{k-1}$.}

\begin{theorem}
\thmmonroecombined\label{thm:monroe-combined}
\end{theorem}

Let us now move to ABC rules offering asymptotically better guarantees than Monroe. As we will see, the examination of such rules requires a more complex combinatorial analysis. We start with PAV:

\newcommand{\thmpav}{The utilitarian guarantee of PAV is between $\frac{1}{2 + \sqrt{k}}$ and $\frac{2}{\sqrt{k}}$;
its representation guarantee is between $\frac{1}{2}$ and $\frac{1}{2} + \frac{1}{4k-2}$.}

\newcommand{\thmslav}{The utilitarian guarantee of SLAV is between $\frac{1}{2 + \sqrt{2k}}$ and $\frac{3}{2\sqrt{k}}$;
its representation guarantee is between $\frac{3}{5}$ and $\frac{2}{3}+\frac{1}{9k-6}$.}

\begin{theorem}\label{thm:guarantees_for_pav}
\thmpav
\end{theorem}
\begin{proof}[Proof of the utilitarian guarantee]
To give a flavor of the proof techniques used in this paper, we show that the utilitarian guarantee of PAV is at least equal to $\frac{1}{2 + \sqrt{k}}$. Consider an approval profile $A$ and a PAV-winning committee $W_{\pav}$; let $\npav = |N(W_{\pav})|$ denote the number of voters who approve some member of $W_{\pav}$. For each $i \in N$ we set $w_i = |W_{\pav} \cap A(i)|$. Let $W_{\av}$ be a committee with the highest AV-score. Without loss of generality, we can assume that $W_{\av} \neq W_{\pav}$. Now, consider a candidate $c \in W_{\av} \setminus W_{\pav}$ with the highest AV-score, and let $n_c = |N(c)|$ denote the number of voters who approve $c$. If we replace a candidate $c' \in W_{\pav}$ with $c$, the PAV-score of $W_{\pav}$ will change by:
\begin{align*}
\Delta(c, c')\ &= \sum_{\substack{i\in N \text{ s.t. } c \in A(i)\\\text{and } c'\notin A(i)}} \frac{1}{w_i+1} - \sum_{\substack{i\in N \text{ s.t. } c' \in A(i)\\\text{and } c\notin A(i)}} \frac{1}{w_i} \\
&= \sum_{i\colon c \in A(i)} \frac{1}{w_i+1} - \sum_{i\colon c' \in A(i)} \frac{1}{w_i} + \sum_{i\colon \{c,c'\} \subseteq A(i)} \frac{1}{w_i}-\frac{1}{w_i+1} 
\\&\geq \sum_{i \in N(c)} \frac{1}{w_i+1} - \sum_{i \in N(c')} \frac{1}{w_i} \text{.}
\end{align*}
Let us now compute the sum
\begin{align*}
\begin{split}
\sum_{c' \in W_{\pav}}\Delta(c, c') &\geq \sum_{c' \in W_{\pav}}\sum_{i \in N(c)} \frac{1}{w_i+1} - \sum_{c' \in W_{\pav}}\sum_{i \in N(c')} \frac{1}{w_i} \\
                             &= k \sum_{i \in N(c)} \frac{1}{w_i+1} - \sum_{i \in N} \underbrace{\sum_{c' \in W_{\pav} \cap A(i)} \frac{1}{w_i}}_{(*)} \\
                             &= k \sum_{i\in N(c)} \frac{1}{w_i+1} - \npav \text{,}
\end{split}
\iffalse \label{eq:sum_delta_pav}\fi
\end{align*}
noting that the expression $(*)$ is 1 if $W_{\pav} \cap A(i)\neq \emptyset$ and 0 otherwise.
We know that for each $c' \in W$ we have $\Delta(c, c') \leq 0$, and thus 
$\sum_{c' \in W_{\pav}}\Delta(c, c')\leq 0$.
Consequently, 
$k \sum_{i \in N(c)} \frac{1}{w_i+1} - \npav \leq 0$ and 
\begin{align}
\sum_{i \in N(c)} \frac{1}{w_i+1} \leq \frac{\npav}{k}.\label{eq:sumw-bound}
\end{align}
Now, recall that the inequality between the harmonic and arithmetic mean says that for all positive values $a_1, \ldots, a_z$ it holds that:
\begin{align*}
\frac{1}{z}\sum_{i = 1}^z a_i \geq \frac{z}{\sum_{i=1}^z \frac{1}{a_i}} \text{.} 
\end{align*}
We reformulate the inequality between the harmonic and arithmetic mean as:
\begin{align}
\sum_{i=1}^z \frac{1}{a_i} \geq \frac{z^2}{\sum_{i = 1}^z a_i} \text{.}\label{eq:inequ-harm-ari} 
\end{align}
By setting $z = |N(c)|$ and $a_i = w_i+1$ we get from inequalities~\eqref{eq:sumw-bound}
and~\eqref{eq:inequ-harm-ari}:
\begin{align*}
\frac{\npav}{k} \geq \sum_{i \in N(c)} \frac{1}{w_i+1} \geq \frac{n_c^2}{\sum_{i \in N(c)} (w_i+1)}
\iffalse = \frac{n_c^2}{\sum_{i\in N(c)} w_i + n_c}
\fi
.
\end{align*}
This can be transformed to:
\begin{align*}
kn_c \leq \frac{\npav \big(n_c + \sum_{i \in N(c)} w_i \big)}{n_c} = \frac{\npav \sum_{i \in N(c)} w_i}{n_c} + \npav\text{.}
\end{align*}
Now, observe that $\score_{\av}(A,W_{\av})\leq \sum_{i \in N} w_i + k n_c$ due to the choice of $c$ and $\score_{\av}(A,W_{\pav})=\sum_{i \in N} w_i$ by definition of $w_i$.
Let us consider two cases. If $\npav \leq n_c \sqrt{k}$, then we observe that:
\begin{align*}
\frac{\score_{\av}(A,W_{\av})}{\score_{\av}(A,W_{\pav})} &\leq \frac{\sum_{i \in N} w_i + k n_c}{\sum_{i \in N} w_i} = 1 + \frac{k n_c}{\sum_{i \in N} w_i} 
\\&\leq 1 + \frac{\frac{\npav \sum_{i \in N(c)} w_i}{n_c} + \npav }{\sum_{i \in N} w_i}\\
&\leq 1 + \frac{\npav }{n_c}\underbrace{\frac{\sum_{i \in N(c)} w_i}{\sum_{i \in N} w_i}}_{\leq 1} + \underbrace{\frac{\npav}{\sum_{i \in N} w_i}}_{\leq 1}
\leq 2 + \frac{\npav}{n_c} \leq \sqrt{k} + 2.
\end{align*}
On the other hand, if $\npav \geq n_c \sqrt{k}$, then:
\begin{align*}
\frac{\score_{\av}(A,W_{\av})}{\score_{\av}(A,W_{\pav})} &\leq \frac{\sum_{i \in N} w_i + k n_c}{\sum_{i \in N} w_i} = 1 + \frac{k n_c}{\sum_{i \in N} w_i} \\
&\leq 1 + \frac{k n_c}{\npav} \leq 1 + \frac{k n_c}{n_c\sqrt{k}} \leq 1 + \sqrt{k}.
\end{align*}
In either case we have that $\frac{\score_{\av}(A,W_{\pav})}{\score_{\av}(A,W_{\av})} \geq \frac{1}{2 + \sqrt{k}}$.
This yields the required lower bound.

The fact that the utilitarian guarantee of PAV is at most equal to $\frac{2}{\lfloor \sqrt{k} \rfloor} - \frac{1}{k}$ will follow from a more general result
in \Cref{sec:propcompromise} (\Cref{thm:proportional_guarantees}).
\end{proof}

While the utilitarian guarantees of SLAV and PAV are of the same order, SLAV gives a considerably better representation guarantee.

\begin{theorem}\label{thm:guarantees_for_slav}
\thmslav
\end{theorem}

However, we will show in \Cref{sec:average_guarantees} through numerical experiments that PAV slightly outperforms SLAV with respect to the utilitarian criterion.

For sequential PAV we can prove a utilitarian guarantee qualitatively similar to the one for PAV.
Concerning its representation guarantee, however, the gap between the lower and upper bounds is large; finding a more accurate estimate remains an interesting open question.

\newcommand{\thmseqpavguaranteecombined}{The utilitarian guarantee of sequential PAV is between $\frac{1}{2\sqrt{k}}$ and $\frac{2}{\lfloor \sqrt{k} \rfloor} - \frac{1}{k}$; its  representation guarantee is between $\frac{1}{\log(k) + 2}$ and $\frac{1}{2} + \frac{1}{4k-2}$.}

\begin{theorem}
\thmseqpavguaranteecombined\label{thm:seq_pav_guarantee-combined} 
\end{theorem}

The following theorem states guarantees for the $p$-Geometric rule.
Let $\W(\cdot)$ denote the Lambert~$\W$ function, i.e., the function that is defined so as to satisfy $z = \W(z) e^{\W(z)}$ for each $z \in \reals$.  
The Lambert $\W$ function increases asymptotically slower than $\log$.

\begin{figure}
\begin{center}
\subcaptionbox{utilitarian guarantee}{
\begin{tikzpicture}[y=3.6cm, x=1.1cm,font=\sffamily, scale=0.8]
\pgfdeclarelayer{bg}
\pgfsetlayers{bg,main}
    \draw (0.8,0) -- coordinate (x axis mid) (6.2,0) node[anchor = north west] {};
        \draw (0.8,0) -- coordinate (y axis mid) (0.8,1.05);
        \foreach \x in {1,...,6}
             \draw (\x,1pt) -- (\x,-3pt)
            node[anchor=north] {\x};
        \foreach \y in {0,.5,1}
             \draw (0.9,\y) -- (0.7,\y) 
                 node[anchor=east] {\y}; 
    \draw (6.3,0.1) node[anchor = north west] {$p$};
    \draw[ultra thick] plot[mark=none] 
        file {avlb.data};
    \draw[ultra thick] plot[mark=none] 
        file {avub.data};        

\end{tikzpicture}
}\quad
\subcaptionbox{representation guarantee}{
\begin{tikzpicture}[y=3.6cm, x=1.1cm,font=\sffamily, scale=0.8]
\pgfdeclarelayer{bg}
\pgfsetlayers{bg,main}
    \draw (0.8,0) -- coordinate (x axis mid) (6.2,0) node[anchor = north west] {};
        \draw (0.8,0) -- coordinate (y axis mid) (0.8,1.05);
        \foreach \x in {1,...,6}
             \draw (\x,1pt) -- (\x,-3pt)
            node[anchor=north west] {\x};
        \foreach \y in {0,.5,1}
             \draw (0.9,\y) -- (0.7,\y) 
                 node[anchor=east] {\y}; 
    \draw (6.3,0.1) node[anchor = north west] {$p$};
    \draw[ultra thick] plot[mark=none] 
        file {cclb.data};
    \draw[ultra thick] plot[mark=none]
        file {ccub.data}; 

\end{tikzpicture}
}
\end{center}
\caption{Guarantees for the $p$-Geometric rule for $k=25$ and varying $p$. In both subfigures, the upper and the lower line depict the upper and lower bound from \Cref{thm:pgeom-combined}.}
\label{fig:pgeom}
\end{figure}

\begin{figure}
\begin{center}
\subcaptionbox{utilitarian guarantee}{
\begin{tikzpicture}[y=3.6cm, x=0.07cm,font=\sffamily, scale=0.8]
\pgfdeclarelayer{bg}
\pgfsetlayers{bg,main}
    \draw (0,0) -- coordinate (x axis mid) (105,0) node[anchor = north west] {};
        \draw (1,0) -- coordinate (y axis mid) (1,1.05);
        \foreach \x in {1,25,50,75,100}
             \draw (\x,1pt) -- (\x,-3pt)
            node[anchor=north] {\x};
        \foreach \y in {0,.5,1}
             \draw (3,\y) -- (-1,\y)
                 node[anchor=east] {\y}; 
    \draw (105,0.1) node[anchor = north west] {$k$};
    \draw[ultra thick] plot[mark=none] 
        file {k-avlb.data};
    \draw[ultra thick] plot[mark=none] 
        file {k-avub.data};        

\end{tikzpicture}
}\quad
\subcaptionbox{representation guarantee}{
\begin{tikzpicture}[y=3.6cm, x=0.07cm,font=\sffamily, scale=0.8]
\pgfdeclarelayer{bg}
\pgfsetlayers{bg,main}
    \draw (0,0) -- coordinate (x axis mid) (105,0) node[anchor = north west] {};
        \draw (1,0) -- coordinate (y axis mid) (1,1.05);
        \foreach \x in {1,25,50,75,100}
             \draw (\x,1pt) -- (\x,-3pt)
            node[anchor=north] {\x};
        \foreach \y in {0,.5,1}
             \draw (3,\y) -- (-1,\y)
                 node[anchor=east] {\y};
    \draw (105,0.1) node[anchor = north west] {$k$};
    \draw[ultra thick] plot[mark=none] 
        file {k-cclb.data};
    \draw[ultra thick] plot[mark=none]
        file {k-ccub.data}; 

\end{tikzpicture}
}
\end{center}
\caption{Guarantees for the $2$-Geometric rule for varying $k$. In both subfigures, the upper and the lower line depict the upper and lower bound from \Cref{thm:pgeom-combined}.}
\label{fig:pgeom-k}
\end{figure}

\newcommand{\thmpgeom}{The utilitarian guarantee of the $p$-Geometric rule is between
\begin{align*}
\frac{\W(k \log(p))} {k \log(p) + \W(k \log(p))} \qquad \text{and} \qquad  \frac{2\W(k \log(p))}{k \log(p)}+ \frac{1}{k},
\end{align*}
its representation guarantee is between 
\begin{align*}\frac{p-1}{p} \qquad \text{and} \qquad \frac{p}{p+\frac{k}{k+2}}.
\end{align*}
}

\begin{theorem}\label{thm:pgeom-combined}\thmpgeom
\end{theorem}

The guarantees of \Cref{thm:pgeom-combined} are visualized in \Cref{fig:pgeom}. We can see that the $p$-Geometric rules, for $p \in [1, \infty)$, form a spectrum connecting AV and CC (with $p \to 1$ we approach AV and with $p \to \infty$ we approach CC): by adjusting the parameter $p$ one can control the tradeoff between diversity and individual excellence. In addition, \Cref{fig:pgeom-k} shows the dependence of the bounds for the 2-Geometric rule on the parameter~$k$; we see that the bounds become more meaningful for growing~$k$.

Finally, we consider seq-Phragm\'{e}n, also a rule aimed at achieving proportional representation. It achieves guarantees similar to PAV.

\newcommand{\thmphrag}{The utilitarian guarantee of seq-Phragm\'{e}n is between $\frac{1}{5\sqrt{k} + 1}$ and $\frac{2}{\lfloor \sqrt{k} \rfloor} - \frac{1}{k}$; its representation guarantee is between $\frac{1}{2}$ and $\frac{1}{2} + \frac{1}{4k-2}$.}

\begin{theorem}\label{thm:phragmen_bounds-combined}
\thmphrag
\end{theorem}

The next result allows for an interesting comparison between seq-Phragm\'{e}n and leximin-Phragm\'{e}n: in terms of utilitarian efficiency leximin-Phragm\'{e}n is significantly worse than its sequential counterpart, whereas the bounds on their representation guarantees are the same.
We see that---from the perspective of our two guarantees---the additional computational complexity of leximin-Phragm\'{e}n does not strengthen its performance (in contrast to PAV and sequential PAV).

\newcommand{\thmoptphrag}{The utilitarian guarantee of leximin-Phragm\'{e}n is $\frac{1}{k}$, and its representation guarantee is between $\frac{1}{2}$ and $\frac{1}{2} + \frac{1}{4k-2}$.}

\begin{theorem}\label{thm:phragmen_opt_bounds-combined}
\thmoptphrag
\end{theorem}

Finally, we show that Minimax Approval Voting is particularly unfavorable with respect to our two criteria.

\newcommand{\thmminimax}{The utilitarian and the representation guarantee of Minimax Approval Voting is $0$.}

\begin{theorem}\label{thm:mnimax_bounds-combined}
\thmminimax
\end{theorem} 

\section{Proportionality as a Compromise?}\label{sec:propcompromise}

In this section we investigate how proportionality relates to diversity and individual excellence, as formalized by the representation and the utilitarian objectives, respectively. In an initial stage of this study we conjectured that proportionality can be characterized as a certain compromise between diversity and individual excellence. However, as we will argue below, proportionality should be rather viewed as a third, independent objective. Indeed, we will construct an ABC rule that is strictly better than any (linearly) proportional rule with regard to both its utilitarian and representation guarantee.  

For each $\alpha \in [0, 1]$ we define $\alpha$-CC-AV as a linear combination of CC and AV. For an approval-based profile $A$, $\alpha$-CC-AV first computes a size-$\lceil \alpha k\rceil$ subset $W_1 \subseteq C$ that maximizes $\score_{\cc}(A, W_1)$ and then returns a set $W$ of size $k$ that (i) is a superset of $W_1$ and (ii) maximizes $\score_{\av}(A, W)$.

\begin{proposition}
The representation guarantee of $\alpha$-CC-AV is at least $\alpha$; its utilitarian guarantee is at least $1 - \alpha - \nicefrac{1}{k}$.
\end{proposition}

\begin{proof}
Consider an approval-based profile $A$, and let $W$ be a committee returned by $\alpha$-CC-AV for $A$. Let $W_{\cc}$ be a committee returned by CC for the same profile. 
Note that $W_1$ achieves at least a $\frac{|W_1|}{k}$-fraction of the CC-score of $W_{\cc}$. Hence
\begin{align*}
\score_{\cc}(A, W_1) \geq \frac{\lceil \alpha k\rceil}{k}\cdot \score_{\cc}(A, W_{\cc}) \geq \alpha\score_{\cc}(A, W_{\cc}) \text{.}
\end{align*}
Thus, also $\score_{\cc}(A, W) \geq \alpha\score_{\cc}(A, W_{\cc})$.
By the same reasoning, we get an analogous result for the utilitarian guarantee.
\end{proof}

Let us now examine what utilitarian and representation guarantees a proportional rule can achieve. We consider a very weak proportionality requirement, called weak proportionality, as studied in the apportionment setting (cf.~\citealp{BaYo82a}).
\iffalse
 in the context of apportionment methods (which are special cases of approval-based multiwinner rules) and is strictly weaker than proportionality axioms typically used in the  context of approval-based multiwinner rules (such as extended and proportional justified representation~\citep{pjr17,justifiedRepresenattion}).
 \fi
Weak proportionality applies only to \emph{party-list profiles}: approval profiles $A$ in which for all voters $i,j\in N$ it holds that either $A(i)=A(j)$ or $A(i)\cap A(j)=\emptyset$; identical voters belong to ``same party''. 

\begin{definition}
We call a profile $A$ a \emph{party-list profile} if for each pair of voters $i, j \in N$ it holds that either $A(i) \cap A(j) = \emptyset$ or that $A(i) = A(j)$. 
We say that $N'\subset N$ is a \emph{party} if it is a maximal subset of $N$ with $A(i) = A(j)$ for $i,j\in N'$.
\end{definition}  

Now, for such profiles, if there exist one or more committees that give each party a number of representatives that is exactly proportional to the party's support, then a proportional rule has to select such a committee.

\begin{definition}
A party-list profile $A$ is $k$-integral if for every party $N'$ it holds that $\frac{k\cdot |N'|}{|N|}$ is an integer and that party $N'$ has at least that many candidates (that is, $|A(i)|\geq \frac{k\cdot |N'|}{|N|}$ for $i\in N'$).
An ABC rule $\calR$ satisfies \emph{weak proportionality} if for every $k \in \naturals$, every $k$-integral party-list profile $A$, and party $N'$ it holds that every winning committee from $\calR(A, k)$ contains exactly $\frac{k\cdot |N'|}{|N|}$ candidates that are approved by $N'$.
\end{definition}  

Weak proportionality is
strictly weaker than other proportionality axioms typically used for ABC rules (such as extended and proportional justified representation~\citep{pjr17,justifiedRepresenattion}), and even weaker than lower and upper quota from the apportionment setting~\citep{BaYo82a}.
Further, \citet{BaYo82a} view this axiom as a minimal requirement for an apportionment method. Thus, any ABC rule that generalizes one of the established apportionment methods also satisfies weak proportionality; this holds in particular for PAV, SLAV, their sequential and reverse-sequential versions, Monroe, and \phragmen's rules \citep{BLS17a}.

\newcommand{\thmproportionalguarantees}{The utilitarian guarantee of a rule satisfying weak proportionality  is at most $\frac{2}{\lfloor \sqrt{k} \rfloor} - \frac{1}{k}$, its representation guarantee is at most $\frac{3}{4}+\frac{3}{8k - 4}$.
}

\begin{theorem}\label{thm:proportional_guarantees}
\thmproportionalguarantees
\end{theorem}

\begin{proof}
Let $\calR$ be a rule that satisfies weak proportionality. 
To see the upper bound on the utilitarian guarantee, let us fix $k$ and consider the following $k$-integral party-list profile $A$ with $n = k$ voters.
The first $\lfloor \sqrt{k} \rfloor$ voters form a party and approve $k$ candidates denoted as $x_1, \ldots, x_k$. All remaining voters $i$, $i > \lfloor \sqrt{k} \rfloor$, form singleton parties and approve a single candidate $y_i$.

Let $W$ and $W_{\av}$ denote the committees returned by $\calR$ and by AV, respectively. Weak proportionality ensures that $y_i \in W$ for each $i > \lfloor \sqrt{k} \rfloor$ and $|W\cap \{x_1, \ldots, x_k\}|=\lfloor \sqrt{k} \rfloor$. Thus, 
\begin{align*}
\score_{\av}(A, W) =  \lfloor \sqrt{k} \rfloor \cdot  \lfloor \sqrt{k} \rfloor + \left(k - \lfloor \sqrt{k} \rfloor\right) \cdot 1 \leq 2k - \lfloor \sqrt{k} \rfloor  \text{.}
\end{align*}
On the other hand, one can observe that $W_{\av} = \{x_1, \ldots, x_k\}$, and so $\score_{\av}(A, W_{\av}) = \lfloor \sqrt{k} \rfloor  \cdot k$. As a result we have:
\begin{align*}
\frac{\score_{\av}(A, W)}{\score_{\av}(A, W_{\av})} \leq \frac{2k - \lfloor \sqrt{k} \rfloor }{\lfloor \sqrt{k} \rfloor \cdot k} = \frac{2}{\lfloor \sqrt{k} \rfloor} - \frac{1}{k} \text{.}
\end{align*}

To prove the upper bound on the representation guarantee, we construct a $k$-integral party-list profile $A$ with $n = 2k$ voters. 
Let us first consider the case that $k$ is even.
The first $k$ voters approve candidates $X = \{x_1, \ldots, x_k\}$, i.e., they form a party. The other $k$ voters are divided into $k$ singleton parties, each approving a different candidate from the set $Y = \{y_1, \ldots, y_{k}\}$. Weak proportionality ensures that $\nicefrac{k}{2}$ candidates need to be chosen from $X$. Thus, the CC-score of a committee selected by $\calR$ is at most equal to $k + \frac{k}{2}$. By selecting one candidate from $X$ and $k-1$ candidates from $Y$ we get a CC-score of $2k - 1$. Thus, the representation guarantee is at most equal to:
\begin{align*}
\frac{k + \frac{k}{2}}{2k - 1} = \frac{3k}{4k - 2}=\frac{3}{4}+\frac{3}{8k - 4} \text{.} 
\end{align*}  

If $k$ is odd, the first $k+1$ voters approve candidates $X = \{x_1, \ldots, x_k\}$, the other $k-1$ voters form singleton parties.
Weak proportionality ensures that $\left\lceil\nicefrac{k}{2}\right\rceil$ candidates need to be chosen from $X$.
Thus, the CC-score of a committee selected by $\calR$ is at most equal to $(k+1) + \left\lfloor\nicefrac{k}{2}\right\rfloor$.
We obtain an upper bound on the representation guarantee of:
\begin{align*}
\frac{(k+1) + \left\lfloor\nicefrac{k}{2}\right\rfloor}{2k - 1} = \frac{(2k+2)+(k-1)}{4k - 2}=\frac{3}{4}-\frac{1}{4k}\leq \frac{3}{4}+\frac{3}{8k - 4} \text{.} 
\end{align*}  
\end{proof}

These bounds yield that $\alpha$-CC-AV for $\alpha=0.76$ achieves a better utilitarian and representation guarantee for all $k\geq 81$ than any rule satisfying weak proportionality. In particular, the utilitarian guarantee of $\alpha$-CC-AV is superior: it guarantees a constant fraction of the optimal AV score.
We conclude that proportional rules can achieve a desirable compromise between diversity and individual excellence (e.g., PAV), but this compromise is not optimal (as we have just seen) and not all proportional rules achieve good utilitarian guarantees (Monroe is no better than CC).

\section{Average Performance: An Experimental Evaluation}\label{sec:average_guarantees}

To complement the theoretical analysis of \Cref{sec:average_guarantees}, we conducted numerical experiments that aim at assessing (average-case) \emph{utilitarian ratios} and \emph{representation ratios} achieved by several voting rules.
These two ratios are per-instance analogues of utilitarian and representation guarantees and are defined as follows:
Given an ABC rule~$\calR$ and a profile $A$, the \emph{utilitarian ratio} and the \emph{representation ratio} are defined as:
\begin{align}
\frac{\min_{W\in \calR(A, k)}\score_{\av}(A,W)}{\max_{W \in \mathcal{S}_k(C)} \score_{\av}(A, W)} \qquad  \text{and} \qquad  \frac{\min_{W\in \calR(A, k)}\score_{\cc}(A,W)}{\max_{W \in \mathcal{S}_k(C)} \score_{\cc}(A, W)} \text{,}\label{eq:ratios}
\end{align}
i.e., we compute the ratio between the optimal committee (according to AV and CC, respectively) and compare it with the committees returned by $\calR$.
In these experiments, we have calculated the utilitarian and representation ratios for real-world and randomly generated profiles and compared them for different ABC rules.\footnote{The Python code for the experiments is available online \citep{martin_lackner_2020_3778972}.}
We have used four datasets: profiles obtained from \texttt{preflib.org}~\citep{conf/aldt/MatteiW13} and three sets of profiles generated via probability distributions (uniformly distributed, Mallows model, and P\'olya's urn model; see details below).

\subsection*{Datasets}
The preflib dataset is based on preferences obtained from \texttt{preflib.org}. 
Since the preflib database contains only very few approval-based datasets, we extracted approval profiles from ranked ballots as follows: for each ranked profile and $i\in\{1,\dots,20\}$, we generated an approval profile assuming that voters approve all candidates that are ranked in the top $i$ positions.
Among these profiles, we restricted our attention to profiles where both the utilitarian ratio of CC and the representation ratio of AV was at most $0.9$ (assuming committee sizes $k\in\{4,\dots,20\}$).
This step excluded profiles where an (almost) perfect compromise between the two criteria exists.
With this method, we obtained a total number of 364 preflib-based instances.

All three random datasets consist of 10,000 profiles with 100 candidates and 50 voters, each,
and the corresponding experiments use a committee size of $k=20$.
Furthermore, as for the preflib instances, we required that in the generated profiles CC has a utilitarian ratio of at most $0.9$ and AV a representation ratio of at most $0.9$.

For the uniform dataset, voters' approval sets are of size $2$--$4$ (chosen uniformly at random); the approval sets of a given size are also chosen uniformly at random.

The Mallows dataset is based on Mallows model \citep{mallows}, a probability distribution for rankings. Given a reference ranking $\sigma$, the probability of a ranking $r$ is
\[\mathbb P(r)=\frac{1}{Z}\phi^{d(r,\sigma)},\]
where $d$ denotes the Kendall-tau distance, $\phi$ is a parameter in $(0,1]$,
and $Z$ is chosen so that the probability mass is $1$.
This model generates rankings that are likely to be similar to the reference ranking $\sigma$; the degree of similarity depends on parameter $\phi$.
For each instance, we chose $\phi$ uniformly at random from the interval $(0,1]$.
Each profile consisting of rankings was transformed to an approval profile by selecting the top-$i$ positions, $i$ being chosen uniformly at random from $\{2,3,4\}$.

Finally, the P\'olya urn model (also refereed to as the P\'{o}lya-Eggenberger urn model) \citep{urnmodels,berg1985urn} was used to generate a further dataset. For each profile, first the size of approval sets was chosen uniformly at random from $\{2,3,4\}$.
The P\'{o}lya urn model is parameterized by a non-negative real $\alpha$, which we chose for each instance uniformly at random from $[0,1]$.
At first, we consider an urn containing all approval sets of the chosen size; if $d$ is the size of approval sets the urn contains $100\choose d$ sets.
The profile is then created by subsequently drawing approval sets from the urn uniformly at random, and then the urn is modified in the following manner.
The drawn approval set is returned to the urn and $\alpha\cdot {100\choose d}$ additional copies of this set is added to the urn.
This procedure is repeated as often as required. The corresponding profiles are likely to contain voters with identical approval sets.

\subsection*{Considered Rules}

We considered the following ABC rules: AV, CC, seq-CC, PAV, seq-PAV, rev-seq-PAV, SLAV, seq-\phragmen, Monroe's rule, MAV, as well as the 1.5-, 2-, and 5-Geometric rule.
We used resolute versions of these rules, i.e., we computed only one winning committee subject to an arbitrary tie-breaking; this simplification allowed us to compute larger committees in reasonable time.
The utilitarian and representation ratio (cf.\ Equation~\ref{eq:ratios}) for a rule $\calR$ and profile $A$ thus simplify to
\begin{align*}
\frac{\score_{\av}(A,\calR(A, k))}{\score_{\av}(A, W_\av)} \qquad  \text{and} \qquad  \frac{\score_{\cc}(A,\calR(A, k))}{\score_{\cc}(A, W_\cc)} \text{,}
\end{align*}
where $W_\av$ and $W_\cc$ are committees with maximum AV and CC score, respectively.

We could not include \phragmen's leximin rule in our experiments as it is too time-consuming to compute for reasonably sized instances. This is mainly due to the non-trivial (lexicographic) tie-breaking required in its computation~\citep{aaai/BrillFJL17-phragmen}, which makes it difficult to solve for ILP solvers. Further research is required to tackle this algorithmic problem.
However, we performed rudimentary experiments for smaller instances and expect the performance of \phragmen's leximin rule to be similar to Monroe (and thus inferior to seq-\phragmen, see below).

\subsection*{Results}
Our results are displayed as boxplots in \Cref{fig:preflib} for the preflib dataset and in \Cref{fig:uniform} for the uniform dataset.
Since the results for the Mallows and the urn dataset are largely similar, we have moved the corresponding plots to the appendix, Section~\ref{app:exp}.
In these plots, the top and bottom of boxes represent the first and third quantiles, the middle red bar shows the median.
The dashed intervals (whiskers) show the range of all values, i.e., the minimum and maximum utilitarian or representation ratio.
In addition, we use our results to rank the rules according to their mean utilitarian and representation ratios (\Cref{tab:util,tab:egal} below, \Cref{tab:util2,tab:egal2} in the appendix).
The differences between the means (and thus the ranking) are largely statistically significant according to a paired t-test with significance level $p=0.01)$.
Those differences that are not statistically significant are marked with brackets in the tables.

\begin{figure}
\includegraphics[width=\textwidth]%
{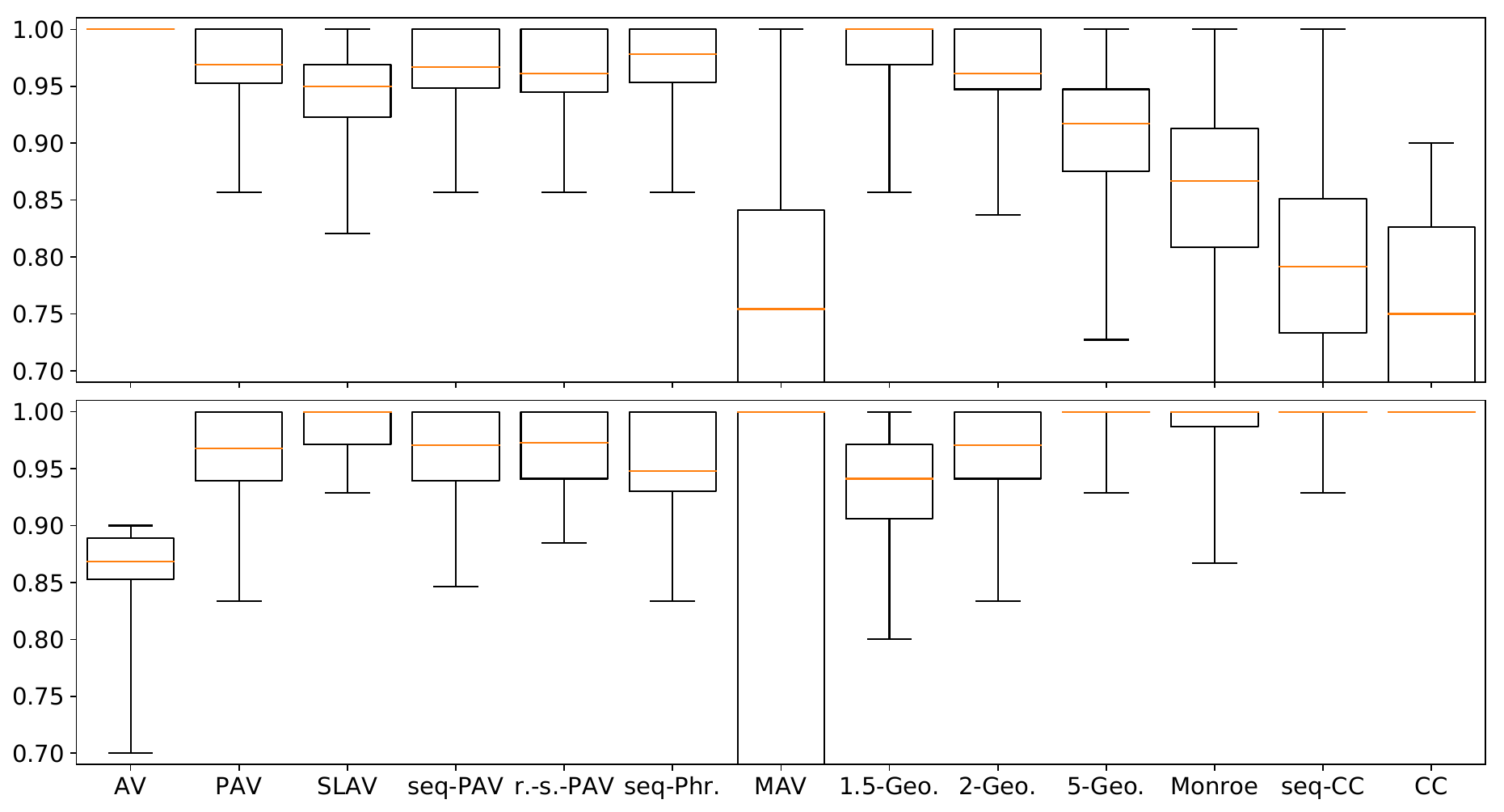}
\caption{Results for the preflib dataset (upper boxplot shows utilitarian ratios, the lower representation ratios).}
\label{fig:preflib}
\end{figure}

\begin{figure}
\includegraphics[width=\textwidth]{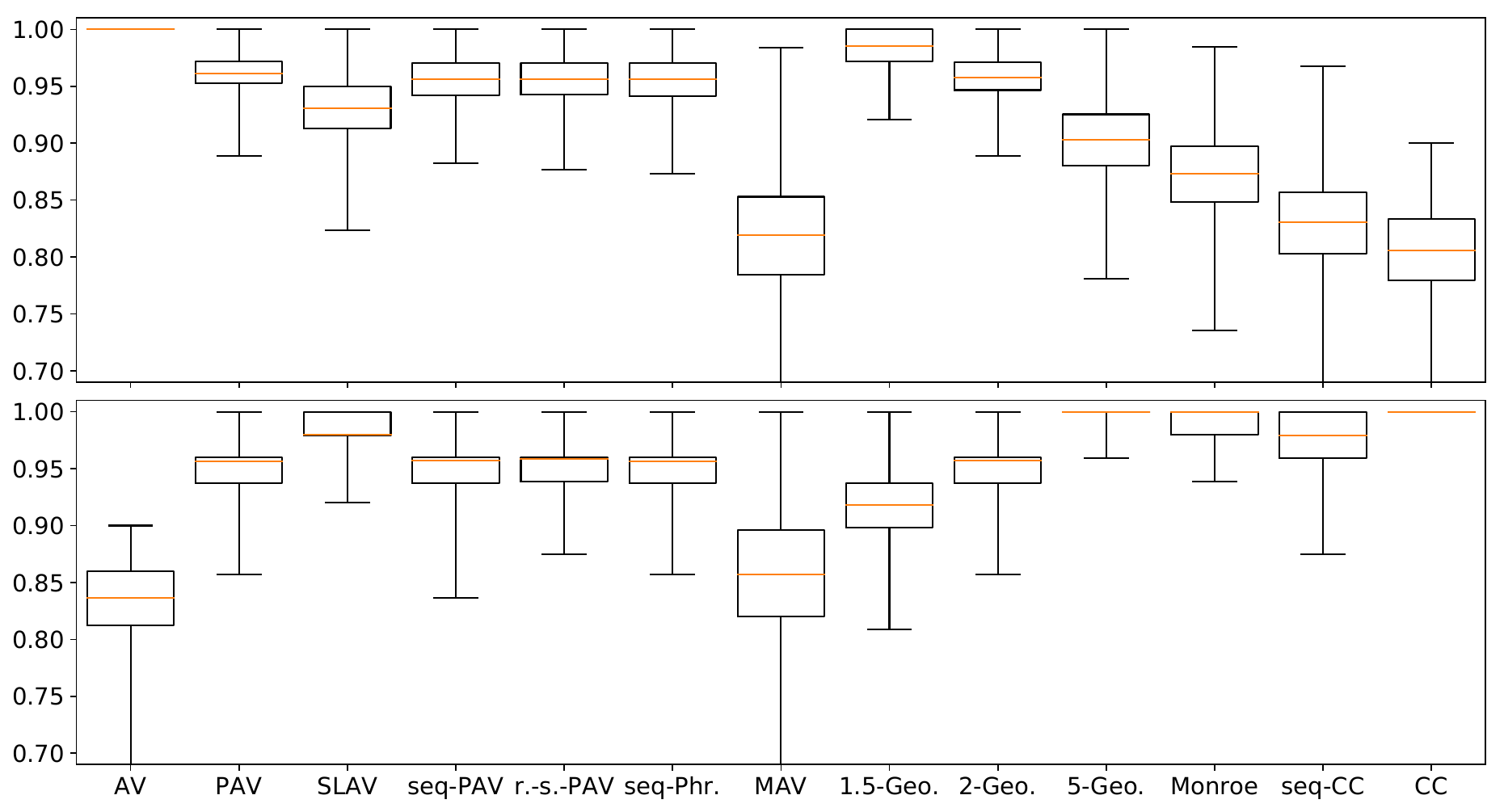}%
\caption{Results for the uniform dataset (upper boxplot shows utilitarian ratios, the lower representation ratios).}
\label{fig:uniform}
\end{figure}

\begin{table}
\begin{center}
\subcaptionbox{preflib dataset}{
\begin{tabular}{clrc}
    &                & mean  & \\
    \toprule
 1. &AV              &1.000& \\
 2. &$1.5$-Geometric &0.982& \\
 3. &seq-\phragmen   &0.973& \\
 4. &PAV             &0.969& \rdelim{]}{2}{1mm}[]\\
 5. &seq-PAV         &0.967& \\
 6. &rev-seq-PAV     &0.963& \rdelim{]}{2}{6mm}[]\\
 7. &$2$-Geometric   &0.961& \\
 8. &SLAV            &0.945& \\
 9. &$5$-Geometric   &0.910& \\
10. &Monroe          &0.861& \\
11. &seq-CC          &0.788& \\
12. &CC              &0.736& \\
13. &MAV             &0.607& \\
\end{tabular}}
\qquad\qquad
\subcaptionbox{uniform dataset}{
\begin{tabular}{clrc}
    &                & mean  & \\
    \toprule
 1. &AV              & 1.000 & \\
 2. &$1.5$-Geometric & 0.984 & \\
 3. &PAV             & 0.962 & \\
 4. &$2$-Geometric   & 0.960 & \\
 5. &seq-PAV         & 0.958 & \rdelim{]}{2}{1mm}[]\\
 6. &rev-seq-PAV     & 0.958 & \\
 7. &seq-\phragmen   & 0.957 & \\
 8. &SLAV            & 0.931 & \\
 9. &$5$-Geometric   & 0.902 & \\
10. &Monroe          & 0.872 & \\
11. &seq-CC          & 0.830 & \\
12. &MAV             & 0.817 & \\
13. &CC              & 0.806 & \\
\end{tabular}}%
\end{center}%
\caption{Mean utilitarian ratios for the preflib and uniform dataset. The differences between pairs of rules are statistically significant (paired t-test, $p=0.01$), unless a pair is marked with a bracket.}\
\label{tab:util}
\end{table}

\begin{table}
\begin{center}
\subcaptionbox{preflib dataset}{
\begin{tabular}{clrc}
    &                & mean  & \\
    \toprule
 1. &CC             &1.000& \\
 2. &$5$-Geometric  &0.997& \\
 3. &seq-CC         &0.991& \rdelim{]}{2}{1mm}[]\\
 4. &Monroe         &0.989& \rdelim{]}{2}{6mm}[]\\
 5. &SLAV           &0.986& \\
 6. &rev-seq-PAV    &0.970& \rdelim{]}{2}{6mm}[]\\
 7. &$2$-Geometric  &0.966& \rdelim{]}{2}{1mm}[]\\
 8. &seq-PAV        &0.963& \rdelim{]}{2}{6mm}[]\\
 9. &PAV            &0.962& \\
10. &seq-\phragmen  &0.951& \\
11. &$1.5$-Geometric&0.938& \\
12. &AV             &0.862& \\
13. &MAV            &0.725& \\
\end{tabular}}
\qquad\qquad
\subcaptionbox{uniform dataset}{
\begin{tabular}{clrc}
    &                & mean  & \\
    \toprule
 1. &CC             & 1.000 & \\
 2. &$5$-Geometric  & 0.998 & \\
 3. &Monroe         & 0.993 & \\
 4. &SLAV           & 0.983 & \\
 5. &seq-CC         & 0.975 & \\
 6. &rev-seq-PAV    & 0.953 & \\
 7. &seq-PAV        & 0.951 & \\
 8. &$2$-Geometric  & 0.950 & \\
 9. &PAV            & 0.949 & \\
10. &seq-\phragmen  & 0.948 & \\
11. &$1.5$-Geometric& 0.917 & \\
12. &MAV            & 0.857 & \\
13. &AV             & 0.836 & \\
\end{tabular}}%
\end{center}%
\caption{Mean representation ratios for the preflib and uniform dataset. The differences between pairs of rules are statistically significant (paired t-test, $p=0.01$), unless a pair is marked with a bracket.}
\label{tab:egal}
\end{table}	

The results for all four datasets are largely similar.
The main difference between datasets is a varying range of ratio values (per rule), while
the relative comparison between rules remains stable.
Also \Cref{tab:util} shows largely the same ranking for both datasets; the same holds for \Cref{tab:egal}.
Further note that the rankings differ only between datasets when mean differences are very small.
Consequently, the following conclusions hold for all considered datasets.

The main conclusion from the experiments is that the classification obtained from worst-case analytical bounds also holds in our (average-case) experiments.
PAV, seq-PAV, rev-seq-PAV, 2-Geometric, and seq-\phragmen{} achieve very high utilitarian ratios, surpassed only by 1.5-Geometric and AV itself.
This is mirrored by our theoretical results as only PAV, seq-PAV, and seq-\phragmen{} achieve a $\Theta(\nicefrac{1}{\sqrt{k}})$ utilitarian guarantee.
The utilitarian ratios of SLAV are slightly below these rules.
For $5$-Geometric, Monroe, seq-CC, and CC we observe significantly lower utilitarian ratios. 
Again, this is reflected in our worst-case analysis: these rules have a $\Theta(\nicefrac{1}{k})$ utilitarian guarantee.
When considering all datasets, the mean utilitarian ratios can be ranked as follows (rules in curly brackets appear in different orders in the datasets):
\begin{align*}
\text{AV             } &>
\text{$1.5$-Geometric  $>$
$\{$PAV,            
$2$-Geometric,  
seq-PAV,        
rev-seq-PAV,    
seq-\phragmen}\}\\
& >  
\text{SLAV         $>$  
$\{5$-Geometric, 
Monroe$\} >$         
$\{$seq-CC,         
MAV,            
CC}\}.     
\end{align*}
This ranking is more fine-grained than the worst-case analysis but still compatible with the more coarse classification obtained from sorting by worst-case bounds.

Now considering the representation ratios, we see almost optimal performance of seq-CC, Monroe, 5-Geometric, and SLAV,
and good performance of PAV, seq-PAV, rev-seq-PAV, seq-\phragmen{}, and 2-Geometric.
Minor variations within these groups depend on the chosen dataset; also note that not all differences are statistically significant (in particular in the preflib dataset).
As before, we can rank the rules according to their mean representation ratios and obtain the following ranking that is consistent with all four datasets:
\begin{align*}
\text{CC             } &>
\text{$5$-Geometric} >  
\{\text{Monroe,         
SLAV,           
seq-CC}\}         >
\{\text{seq-PAV, rev-seq-PAV,
$2$-Geometric}\}\\ &>  
\text{PAV} >            
\text{seq-\phragmen} >  
\text{$1.5$-Geometric} >
\{\text{MAV,            
AV}\}.           
\end{align*}

When looking at the three Geometric rules considered here, we see in \Cref{fig:preflib,fig:uniform} the transition from AV to CC as our theoretical findings predict (cf. \Cref{fig:pgeom}): 1.5-Geometric is close to AV, whereas 5-Geometric resembles CC; 2-Geometric performs very similarly to PAV.

Our results indicate that PAV, seq-PAV, rev-seq-PAV, seq-\phragmen{}, and 2-Geometric provide the best compromise between AV and CC.
Note, however, among those seq-PAV, rev-seq-PAV, and seq-\phragmen{} are computable in polynomial time---which does not translate to inferior performance.
If a bit more emphasis is put on the representation criterion, then SLAV appears to be the best choice. 
We did not investigate seq-SLAV, but expect a similar behavior (as it is the case for PAV and seq-PAV).
Also note that seq-PAV and rev-seq-PAV are virtually indistinguishable in these experiments.

It is important to note that proportionality is not the explaining factor for the strong performance of the aforementioned rules: while PAV satisfies the strong proportionality axiom ``extended justified representation'' \citep{justifiedRepresenattion}, seq-\phragmen{} satisfies the weaker variant ``proportional justified representation'' \citep{pjr17}, and seq-PAV and rev-seq-PAV only satisfy the even weaker ``D'Hondt proportionality'' axiom \citep{BLS17a}, 2-Geometric is not proportional at all.
This further corroborates our result from \Cref{sec:propcompromise} that proportionality should not be viewed as an optimal compromise between AV and CC.

Sometimes the goal is not to balance the performance with respect to utilitarian and representation ratios, but rather to put emphasis on one of them.
If the desiderata is a high utilitarian efficiency, our experiments recommend a $p$-Geometric rule with $p<2$.
On the contrary, if much weight is put on representing the voters, then CC, seq-CC, Monroe, and $p$-Geometric rules for $p\geq 5$ shine. These rules, however, severely lack in achieving utilitarian welfare.
Thus, as mentioned before, a recommendable choice is SLAV, which achieves very good representation ratios and still very solid utilitarian ratios (much better than the aforementioned rules).

\section{A Pareto Efficiency Axiom}\label{sec:efficiency_axiom}

In this section, we provide a complementary axiomatic analysis concerning the principle of individual excellence.
This analysis should show that it is difficult to find axioms that can separate those rules that fulfill this principle rather well and those that do not.
Note that this is exactly what we achieved with the concept of utilitarian guarantee (i.e., with a non-axiomatic approach).

We approach this problem by formulating a weak version of Pareto efficiency, which states that only Pareto optimal committees should be selected.
Pareto optimality is usually defined based on a notion of utility; we stick with our general assumption that voters derive utility from the number of approved candidates in the committee.\footnote{We note that this specific variant of Pareto efficiency is the same as Pareto optimality for dichotomous preferences with respect to the responsive set extension as studied by \citet{aziz2020computing}.}

\todo{change definition to survey-def.}

\begin{definition}
Consider a committee size $k \in \naturals$, two committees $W_1, W_2 \in \mathcal{S}_k(C)$ and an approval profile $A \in \calA$. We say that \emph{$W_1$ dominates $W_2$} in a $A$ if for each voter $i \in N$ we have that $|W_1 \cap A(i)| \geq |W_2 \cap A(i)|$, and if there exists a voter $j$ such that  $|W_1 \cap A(j)| > |W_2 \cap A(j)|$.
An undominated committee is \emph{Pareto optimal}.

An ABC rule $\calR$ satisfies \emph{Pareto efficiency} if for each profile $A \in \calA$ and each committee size~$k$ there exists no committee $W \in \mathcal{S}_k(C)$ that dominates each committee in $\calR(A, k)$.
\end{definition}

This axiom is rather weak since it requires that there must not exist a committee that dominates \emph{all} winning committee. An alternative would be to consider an axiom which requires winning committees are Pareto optimal. We chose this weaker variant as we were looking for a minimal axiom that could capture the idea of individual excellence.

As Pareto efficiency is often considered to be a minimal requirement and the chosen notion of utility aligns with the principle of individual excellence, one could expect that this axiom is satisfied by voting rules that fulfill this principle.
However, our analysis shows that many sensible rules do not satisfy this basic axiom, and those that satisfy it include individually excellent rules as well as diverse rules.

First, we make the rather surprising observation that seq-\phragmen{} does not satisfy Pareto efficiency.

\begin{example}\label{ex:seq-phrag-pareto}
Consider the set of 36 voters, and five candidates, $c_1, \ldots, c_5$. By $N(c)$ we denote the set of voters who approve $c$. Assume that:
\begin{align*}
& N(c_1) = \{1, \ldots, 20\}; \quad && N(c_2) = \{11, \ldots, 28\}; \quad && N(c_3) = \{1, \ldots, 10, 29, \ldots, 36\}; \\
& N(c_4) = \{21, \ldots, 36\}; \quad && N(c_5) = \{1, \ldots, 19\}. &&
\end{align*}
Phragm\'{e}n's sequential rule will select $c_1$ first, $c_4$ second, and $c_5$ third, yet committee $\{c_1, c_4, c_5\}$ is dominated by $\{c_1, c_2, c_3\}$. %
\end{example}

We note that the violation of Pareto efficiency is not an artifact of the rule being sequential (and so, in some sense ``suboptimal''). 
The following example shows that also leximin-Phragm\'{e}n does not satisfy Pareto efficiency. In addition, the same example shows that Monroe's rule does not satisfy it either.

\begin{example}\label{ex:leximin-phrag-pareto}
Consider 24 voters, and four candidates, $c_1, \ldots, c_4$, with the following preferences:
\begin{align*}
& N(c_1) = \{3, \ldots, 22\}; \quad && N(c_2) = \{1, 2, 23, 24\};  \\
& N(c_3) = \{2, \ldots, 12\}; \quad && N(c_4) = \{13, \ldots, 23\}.
\end{align*}
Phragm\'{e}n's leximin rule and Monroe's rule select $\{c_3, c_4\}$, which is dominated by~$\{c_1, c_2\}$.
\end{example}

Greedy Monroe, seq-CC, and seq-PAV do not satisfy Pareto efficiency either. Intuitively, this is due to their sequential nature; the corresponding example can be found in Appendix~\ref{app:axiom}. 
All the remaining rules that we consider satisfy Pareto efficiency.

\newcommand{\propositionParetoSatisfiedAVCCPAVMAV}{AV, CC, PAV, $p$-Geometric, and MAV satisfy Pareto efficiency.}

\begin{proposition}\label{prop:AV_CC_PAV_satisfy_eff}
\propositionParetoSatisfiedAVCCPAVMAV
\end{proposition}

Moreover, it is easy to see that all Thiele methods satisfy Pareto efficiency.
Note that winning committees according to CC and MAV may contain dominated committees (this does not violate the Pareto efficiency axiom as stated above). For AV, PAV, and $p$-Geometric the stronger property holds that any winning committee must be Pareto optimal.

Our conclusion from the above analysis can be summarized as follows. First, we observe that some very sensible rules do not satisfy this basic form of Pareto efficiency---in particular, none of the two variants of the \phragmen{} rule does. Second, we find this axiom is not really related to the utilitarian efficiency of multiwinner rules, as it is satisfied, e.g., by CC but not by seq-\phragmen.
We believe that this analysis
illustrates the problem of axiomatically separating voting rules with strong and weak utilitarian guarantees, e.g., those with a utilitarian guarantee of $\Theta(\nicefrac{1}{\sqrt{k}})$ (such as AV, PAV, and seq-\phragmen) and those with a utilitarian guarantee of $\Theta(\nicefrac{1}{{k}})$ (such as Monroe and CC).
In particular, an axiom that achieves such a separation cannot rely on the notion of Pareto optimality.

\section{Conclusions and Directions for Future Research}\label{sec:concl}

In this work, we assess qualities of multiwinner rules with respect to a utilitarian and a representation-focused criteria. Our results can help to understand the landscape of multiwinner rules, specifically how they behave with respect to these two contradictory goals.

We have mentioned briefly in the introduction that the representation guarantee can be viewed as an egalitarian criterion. We want to discuss this claim and discuss other approaches to measure how egalitarian an ABC rule is.
A classic egalitarian criterion is to maximize the utility of the least-satisfied voter~\citep{moulinAxioms}, or, in the setting of ABC rules, to maximize the number of approved candidates in the committee for the voter with fewest of them \citep{AFGST-egalitarian}. If we took this criterion, then none of the rules that we consider in this paper---except CC---would have a non-zero worst-case guarantee.\footnote{Assume $k$ is even, and consider a profile with $n-1$ voters approving candidates $C_1 = \{c_1, \ldots, c_{2k}\}$, and one voter approving $C_2 = \{c_{2k+1}, \ldots, c_{3k}\}$. If the number of voters $n$ is sufficiently large, then all the rules that we consider in this paper (all the rules listed in \Cref{sec:prelim}) except for Approval Chamberlin--Courant rule, would pick $k$ candidates from $C_1$. Therefore, in this profile, the utility of the least-satisfied voter would be equal to zero. Selecting $\nicefrac{k}{2}$ candidates from $C_1$ and $\nicefrac{k}{2}$ from $C_2$, would result in the utility of the worst-off voter equal to $\nicefrac{k}{2}$. Thus, according to this criterion, the worst-case ratio of the quality of the elected committee divided by the quality of the optimal committee is for almost all known rules equal to zero.} 
Hence, the analysis of this criterion would not lead to meaningful theoretical results as it does not help to distinguish standard ABC rules.

Another natural choice for an egalitarian criterion is to maximize the leximin welfare of voters \citep{moulinAxioms}, i.e., to first maximize the number of voters with at least one approved committee member, then to break possible ties between optimal committees by maximizing the number of voters with at least two approved committee members, etc.
If we assume that we cannot guarantee each voter an approved candidate in the committee, this criterion is exactly the same as our representation guarantee.
Therefore, we think it is justified to view the representation guarantee as a somewhat egalitarian measure.
Note that the leximin criterion itself is not suitable as a basis for a quantitative guarantee as it does not rank committees by scores.

\medskip

Our work can be extended in several directions.
First, we have focused on approval-based multiwinner rules---a natural next step is to perform a similar analysis for multiwinner rules that take ranked ballots as input.
Second, we have chosen AV and CC as extreme notions that represent diversity and individual excellence. Another natural approach would be to look at quantitative notions of proportionality: can such an approach encompass and extend the ``proportionality hierarchy'' of D'Hondt proportionality \citep{lac-sko:t:abc-approval-multiwinner,BLS17a}, justified representation \citep{justifiedRepresenattion}, proportional justified representation \citep{pjr17}, and extended justified representation \citep{justifiedRepresenattion}?
Recent work on voting rules approximating the core property~\citep{fain2018fair,cheng2019group,peters2019proportionality} and on measuring the average satisfaction of cohesive groups~\citep{pjr17,AEHLSS18,proprank,skowron:prop-degree} can be seen as steps in that direction.
Third, our axiomatic analysis raises important questions: Does there exist a natural variant of the \phragmen{} rule that is Pareto efficient? And is there a meaningful axiomatic property that separates ABC rules with strong utilitarian or representation guarantees from those with weak guarantees?

As we have noted in the introduction, we assume that there is a linear relation between a voter's utility and the number of approved committee members---this directly leads to the definition of the utilitarian guarantee that is adopted in this paper. 
However, there are also other well-grounded ways in which the utility of a voter can be defined. For example, one can assume diminishing marginal utilities or one can consider utility functions that count not only approved committee members but also non-approved non-selected candidates.
The latter kind of utility functions is used in the definitions of ordered weighted average (OWA) rules that span the spectrum between AV and MAV~\citep{ABLMR-hamingDistances}, in the definition of dissatisfaction counting rules~\citep{lac-sko:t:multiwinner-strategyproofness}, and in the context of distance-based belief merging \citep{KoniecznyLM02,KoniecznyP02,konieczny2005propositional,aaai/propbm}. 
When viewed from this perspective, the Approval Chamberlin--Courant rule could be interpreted as maximizing utilitarian social welfare when assuming a different utility function: the utility of a voter equals one if she is represented by at least one member of the elected committee, and zero otherwise. This is not the perspective we have taken in this paper; we view the difference between AV and CC (and other Thiele rules) in the way in which preferences of voters (and thus their utility functions) are aggregated rather than in the way the utility of an individual voter is assumed to be. Analyzing the utilitarian welfare of multiwinner rules assuming different kinds of utility functions (such as those mentioned here) is an interesting direction for future research.

\section*{Acknowledgements}

Martin Lackner was supported by the European Research Council (ERC) under grant number 639945 (ACCORD) and by the Austrian Science Foundation FWF, grant P31890. 
Piotr Skowron was supported by a Humboldt Research Fellowship for Postdoctoral Researchers and by the Foundation for Polish Science within the Homing programme (Project title: "Normative Comparison of Multiwinner Election Rules").
The authors thank Piotr Faliszewski for valuable discussions while visiting AGH University of Science and Technology, Krakow.

\bibliographystyle{abbrvnat}
\bibliography{main}

\appendix

\section{Proof Details from \Cref{sec:worst_case_guarantees}}

\begin{reptheorem}{thm:monroe-combined}
\thmmonroecombined
\end{reptheorem}

\begin{proof}[Proof of the utilitarian guarantee]
First, let us consider the utilitarian guarantee of the Greedy Monroe rule. To see the lower bound of $\nicefrac 1 k$, let $A$ be an approval profile and let $\bar c$ denote the candidate who is approved by most voters.
For the sake of clarity we assume that $k$ divides $n$; the proof can be generalized to hold for arbitrary $n$.
Clearly, for any committee $W$ it holds that $\score_{\av}(A, W) \leq k |N(\bar c)|$. If $|N(\bar c)| \leq \frac{n}{k}$, then the Greedy Monroe rule in the first step will select $\bar c$. Otherwise, it will select some candidate approved by at least $\frac{n}{k}$ voters, and will remove $\frac{n}{k}$ of them from $A$. By a similar reasoning we can infer that in the second step the rule will pick a candidate who is approved by at least $\min\left(\frac{n}{k}, |N(\bar c)| - \frac{n}{k}\right)$ voters; and in general, that in the $i$-th step the rule will pick the candidate who is approved by at least $\min\left(\frac{n}{k}, |N(\bar c)| - \frac{n(i-1)}{k}\right)$ voters. As a result, we infer that number of voters that have at least one approved candidate in the chosen committee is at least
\[\sum_{i=1}^k \min\left(\frac{n}{k}, |N(\bar c)| - \frac{n(i-1)}{k}\right)=|N(\bar c)|.\]
Hence the utilitarian guarantee of Greedy Monroe is at least $\nicefrac 1 k$. 

To see that the same lower bound holds for Monroe's rule, we distinguish two cases; let $W$ be a winning committee.
If $\bar c\in W$, then $\score_{CC}(A,W)\geq |N(\bar c)|$ and we are done.
If $\bar c\notin W$ and $\score_{CC}(A,W)< |N(\bar c)|$, then there is a committee with a higher Monroe-score that contains $\bar c$; a contradiction.

Now, consider the following instance witnessing that the utilitarian guarantee of Greedy Monroe is at most $\frac{1}{k}$.
Let $n=k\cdot (x+1)$ and let $A$ be a profile with $n$ voters.
Let $W\subseteq C$ with $|W|=k$ and $c_1,\dots,c_{k}\notin W$.
We define profile $A$ as follows: 
we have $x$ voters that approve $W\cup\{c_1\}$ and one voter that approves only $\{c_1\}$,
we have $x$ voters that approve $W\cup\{c_2\}$ and one voter that approves only $\{c_2\}$, etc.
This defines in total $k\cdot (x+1)$ voters.
AV selects the committee $W$ with an AV-score of $xk^2$;
Greedy Monroe selects the committee $\{c_1\dots,c_{k}\}$ with an AV-score of $(x+1)k$.
We have a ratio of 
$
\frac{(x+1)}{xk}
$, 
which converges to $\frac{1}{k}$ for $x\to\infty$.
The same instance shows that the utilitarian guarantee of Monroe's rule is at most $\frac{1}{k}$.
\end{proof}

\begin{proof}[Proof of the representation guarantee]
We move on to proving bounds for the representation guarantee.
First, for the sake of contradiction let us assume that there exists a profile $A$ where the representation guarantee of Greedy Monroe is below $\frac{1}{2}$. Let $W_{\cc}$ and $W_{M}$ be the committees winning in $A$ according to CC and Greedy Monroe, respectively. Let $\phi$ be an assignment of the voters to the committee members obtained during the construction of $W_{M}$; we say that a voter is represented if it is assigned to a member of $W_{M}$ who she approves of. Since $\score_{\cc}(A, W_{M}) < \frac{1}{2} \cdot \score_{\cc}(A, W_{\cc})$, by the pigeonhole principle we infer that there exists a candidate $c \in W_{\cc} \setminus W_{M}$ who is approved by $x$ unrepresented voters, where:
\begin{align*}
x &\geq \frac{\score_{\cc}(A, W_{\cc}) - \score_{\cc}(A, W_{M})}{k} \geq \frac{2 \score_{\cc}(A, W_{M}) - \score_{\cc}(A, W_{M})}{k} = \frac{\score_{\cc}(A, W_{M})}{k}.
\end{align*}
Similarly, by the pigeonhole principle we can infer that there exists a candidate $c' \in W_{M}$ who is represented by at most $\frac{\score_{\cc}(A, W_{M})}{k}$ voters. Thus, Greedy Monroe would select $c$ rather than $c'$, a contradiction. A similar argument can be made to show that the representation guarantee of Monroe's rule is $\geq \frac{1}{2}$. 
 
Now, consider the following approval profile. There are $2k+1$ candidates, $c_1, \ldots, c_{2k+1}$, and $2k$ disjoint equal-size groups of voters, $N_1, \ldots, N_{2k}$. For each $i \in [2k]$, candidate $c_i$ is approved by all voters from $N_i$. Candidate $c_{2k+1}$ is approved by all voters from $N_1 \cup \cdots \cup N_k$. One of the winning committees according to Monroe and Greedy Monroe is $\{c_1, \ldots, c_{k-2}, c_{k+1}, c_{2k+1}\}$, which has a CC-score of $\frac{n}{k} + (k-1)\frac{n}{2k}$. On the other hand, $\{c_{k+1}, \ldots, c_{2k-1}, c_{2k+1}\}$ has a CC-score of $n - \frac{n}{k}$. Thus, the representation guarantee of Monroe and Greedy Monroe is at most:
\begin{align*}
\frac{\frac{n}{k} + \frac{n(k-1)}{2k}}{n - \frac{n}{k}} = \frac{\frac{k + 1}{2k}}{\frac{k-1}{k}} = \frac{k + 1}{2k-2} = \frac{1}{2} + \frac{1}{k-1} \text{.}
\end{align*}
This completes the proof.
\end{proof}

\begin{reptheorem}{thm:guarantees_for_pav}
\thmpav
\end{reptheorem}

\begin{proof}[Proof of the utilitarian guarantee] See main text.
\end{proof}
\begin{proof}[Proof of the representation guarantee]
We first prove a lower bound of $\nicefrac{1}{2}$ for the representation guarantee of PAV.
Consider an approval-based profile $A$ and a PAV winning committee $W_{\pav}$. 
For each voter $i \in N$ we set $w_i = |W_{\pav} \cap A(i)|$. Let $W_{\cc}$ be a committee winning according to the Chamberlin--Courant rule.
For each two candidates, $c \in W_{\pav}$ and $c' \in W_{\cc}$, let $\Delta(c', c)$ denote the change of the PAV-score of $W_{\pav}$ due to replacing $c$ with $c'$. 
Recall from the proof of PAV's utilitarian guarantees that
if we replace a candidate $c' \in W_{\pav}$ with $c$, the PAV-score of $W_{\pav}$ will change by:
\begin{align*}
\Delta(c', c) \geq \sum_{i \in N(c')} \frac{1}{w_i+1} - \sum_{i \in N(c)} \frac{1}{w_i} \text{.}
\end{align*}
Let us now consider an arbitrary bijection $\tau\colon W_{\pav} \to W_{\cc}$, matching members of $W_{\pav}$ with the members of $W_{\cc}$.  
We compute the sum:
\begin{align*}
\sum_{c \in W_{\pav}}\Delta(\tau(c), c) &\geq \sum_{c' \in W_{\cc}}\sum_{i \in N(c')} \frac{1}{w_i+1} - \sum_{c \in W_{\pav}}\sum_{i \in N(c)} \frac{1}{w_i} \\
                             &= \sum_{i \in {N({W_{\cc}})}} \underbrace{\sum_{c' \in W_{\cc} \cap A(i)} \frac{1}{w_i+1}}_{\geq \frac{1}{w_i+1}} - \sum_{i \in N({W_{\pav}})} \underbrace{\sum_{c \in W_{\pav} \cap A(i)} \frac{1}{w_i}}_{=1} \\
                             &\geq \sum_{i \in {N({W_{\cc}})}} \frac{1}{w_i+1} - |N({W_{\pav}})| \geq \sum_{i \in {N({W_{\cc}})\setminus N({W_{\pav}})}} 1  - |N({W_{\pav}})|  \\
                             &\geq |N({W_{\cc}}) \setminus N({W_{\pav}})| - |N({W_{\pav}})|\\
                             &\geq |N({W_{\cc}})| - |N({W_{\pav}})| - |N({W_{\pav}})| \\ 
                             &= |N({W_{\cc}})| - 2|N({W_{\pav}})| \text{.}
\end{align*}
Since $W_{\pav}$ is an PAV-optimal committee, we know that for each $c \in W_{\pav}$, it holds that $\Delta(\tau(c), c) \leq 0$. Consequently, $\sum_{c \in W_{\pav}}\Delta(\tau(c), c) \leq 0$, and so we get that $|N_{W_{\cc}}| - 2|N_{W_{\pav}}| \leq 0$, 
Consequently, we get that $|N_{W_{\pav}}| \geq \frac{|N_{W_{\cc}}|}{2}$, which shows that the representation guarantee of PAV is at least equal to $\nicefrac{1}{2}$.

Now, we will prove the upper bound using the following construction. Let $n$, the number of voters, be divisible by $2k$. 
The set of candidates is $X\cup Y$ with $X = \{x_1, \ldots, x_k\}$ and $Y = \{y_1, \ldots, y_k\}$.
There are $\nicefrac{n}{2}$ voters who approve $X$. Further, for each $i\in[k]$, there are $\frac{n}{2k}$ voters who approve candidate $y_i$. All committees that contain at least $k-1$ candidates from $X$ are winning according to PAV, among them $X$ itself. Committee $X$ has a CC-score of $\nicefrac{n}{2}$. The optimal CC committee consists of a single candidate from $X$ and $(k-1)$ candidates from $Y$---this would give a CC-score of $\frac{n}{2} + (k-1)\cdot \frac{n}{2k} = n \cdot \frac{2k-1}{2k}$. Thus, the representation guarantee of PAV is at most equal to:
\begin{align*}
\frac{2k}{4k-2} = \frac{1}{2} + \frac{1}{4k-2} \text{.}
\end{align*}
This completes the proof.
\end{proof}

\begin{reptheorem}{thm:guarantees_for_slav}
\thmslav
\end{reptheorem}

\begin{proof}[Proof of the utilitarian guarantee] 
Here, we only explain how the differences in comparison to the proof for PAV.
For the lower bound we get that:
\begin{align*}
\Delta&(c, c') \geq \sum_{i \in N(c)} \frac{1}{2w_i+1} - \sum_{i \in N(c')} \frac{1}{2w_i-1} \text{.}
\end{align*}
And so:
\begin{align*}
\sum_{c' \in W_{\slav}}\Delta(c, c') \geq k \sum_{i\in N(c)} \frac{1}{2w_i+1} - \npav \text{.}
\end{align*}
Consequently, $k \sum_{i \in N(c)} \frac{1}{2w_i+1} - \npav \leq 0$ and $\sum_{i \in N(c)} \frac{1}{2w_i+1} \leq \frac{\npav}{k}$. 
We use the inequality between the harmonic and arithmetic mean:
\begin{align*}
\frac{\npav}{k} \geq \sum_{i \in N(c)} \frac{1}{2w_i+1} \geq \frac{n_c^2}{\sum_{i \in N(c)} (2w_i+1)}.
\end{align*}
From this it follows that:
\begin{align*}
kn_c \leq \frac{\npav \big(\sum_{i \in N(c)} 2w_i + n_c\big)}{n_c} = \frac{2\npav \sum_{i \in N(c)} w_i}{n_c} + \npav \text{.}
\end{align*}
Now, let us consider two cases. If $\npav \leq n_c \sqrt{\nicefrac{k}{2}}$, then we observe that:
\begin{align*}
&\frac{\score_{\av}(A,W_{\av})}{\score_{\av}(A,W_{\slav})} \leq \frac{\sum_{i \in N} w_i + k n_c}{\sum_{i \in N} w_i} = 1 + \frac{k n_c}{\sum_{i \in N} w_i} 
\\&\leq 1 + \frac{\frac{2\npav \sum_{i \in N(c)} w_i}{n_c} + \npav }{\sum_{i \in N} w_i}\leq 2 + \frac{2\npav}{n_c} \leq \sqrt{2k} + 2.
\end{align*}
If $\npav \geq n_c \sqrt{\nicefrac{k}{2}}$, then similarly as in the proof for PAV:
\begin{align*}
\frac{\score_{\av}(A,W_{\av})}{\score_{\av}(A,W_{\slav})} &\leq \frac{\sum_{i \in N} w_i + k n_c}{\sum_{i \in N} w_i} = 1 + \frac{k n_c}{\sum_{i \in N} w_i} \\
&\leq 1 + \nicefrac{k n_c}{\npav} \leq 1 + \sqrt{2k}.
\end{align*}
In either case we have that $\frac{\score_{\av}(A,W_{\slav})}{\score_{\av}(A,W_{\av})} \geq \frac{1}{2 + \sqrt{2k}}$.
This yields the required lower bound.

We adapt the construction used in the proof of \Cref{thm:proportional_guarantees} to obtain a smaller upper bound ($\frac{3}{2\lfloor \sqrt{k} \rfloor}$ instead of $\frac{2}{\lfloor \sqrt{k} \rfloor}-\frac 1 k$).
We change the instance so that the first group consists of $2\lfloor \sqrt k \rfloor$ voters (these approve $\{x_1,\dots,x_k\}$);
the remaining $k-2\lfloor \sqrt k \rfloor$ voters independently approve single candidates (voter $i$ approves $y_i$).
Towards a contradiction, assume that the first group has more than $\lfloor \sqrt k \rfloor$ approved candidates in the winning committee.
For this committee, if we exchange one $x$ candidate with a $y$ candidate, the SLAV score for the first group decreases by at most
\begin{align*}
2\lfloor \sqrt k \rfloor\cdot \frac{1}{2(\lfloor \sqrt k \rfloor+1)-1}< 2\lfloor \sqrt k \rfloor\cdot \frac{1}{2\lfloor \sqrt k \rfloor}=1,
\end{align*}
but the SLAV score of the voter approving the $y$ candidate would increase by one.
This is a contradiction that to the assumption that the first group has more than $\lfloor \sqrt k \rfloor$ approved candidates in the winning committee.
Now let $W$ be a winning committee. As in the proof of \Cref{thm:proportional_guarantees}, we see that
\begin{align*}
\frac{\score_{\av}(A, W)}{\score_{\av}(A, W_{\av})} \leq \frac{2 \lfloor \sqrt{k} \rfloor \cdot \lfloor \sqrt{k} \rfloor + (k-\lfloor \sqrt{k} \rfloor)\cdot 1 }{2 \lfloor \sqrt{k} \rfloor \cdot k} \leq \frac{1}{\lfloor \sqrt{k} \rfloor} + \frac{1}{2\lfloor \sqrt{k} \rfloor}=\frac{3}{2\lfloor \sqrt{k} \rfloor}\text{.}
\end{align*}

\end{proof}

\begin{proof}[Proof of the representation guarantee]
As before, we will explain how to adapt the proof for PAV. For the lower bound, we now have:
\begin{align*}
\sum_{c \in W_{\slav}}\Delta(\tau(c), c) &\geq \sum_{c' \in W_{\cc}}\sum_{i \in N(c')} \frac{1}{2w_i+1} - \sum_{c \in W_{\slav}}\sum_{i \in N(c)} \frac{1}{2w_i-1} \\
                             &= \sum_{i \in {N({W_{\cc}})}} \underbrace{\sum_{c' \in W_{\cc} \cap A(i)} \frac{1}{2w_i+1}}_{\geq \frac{1}{2w_i+1}} - \sum_{i \in N({W_{\slav}})} \underbrace{\sum_{c \in W_{\slav} \cap A(i)} \frac{1}{2w_i-1}}_{= \frac{w_i}{2w_i - 1} \leq 1} \\
                             &\geq \sum_{i \in {N({W_{\cc}}) \setminus N({W_{\slav}})}} \frac{1}{2w_i+1} + \sum_{i \in {N({W_{\cc}}) \cap N({W_{\slav}})}} \left(\frac{1}{2w_i+1} - \frac{w_i}{2w_i - 1}\right)\\
                             &\qquad - |N({W_{\slav}}) \setminus N({W_{\cc}})| \text{.}
\end{align*}
Now, observe that $w_i = 0$ for each $i \in N({W_{\cc}}) \setminus N({W_{\slav}})$ and that $w_i \geq 1$ for each $i \in N({W_{\cc}}) \cap N({W_{\slav}})$. Further, through simple arithmetic calculations we can show that for $w_i \geq 1$ it holds that $\frac{1}{2w_i+1} - \frac{w_i}{2w_i - 1} \geq -\frac{2}{3}$. Thus, we can continue our calculations:

\begin{align*}
\sum_{c \in W_{\slav}}\Delta(\tau(c), c) &\geq |N({W_{\cc}}) \setminus N({W_{\slav}})| - \frac{2}{3} |N({W_{\slav}}) \cap N({W_{\cc}})|- |N({W_{\slav}}) \setminus N({W_{\cc}})|\\
                                         &= |N({W_{\cc}})| - \frac{2}{3} |N({W_{\slav}}) \cap N({W_{\cc}})|- |N({W_{\slav}})| \\
                                         &\geq |N({W_{\cc}})| - \frac{5}{3} |N({W_{\slav}})| \text{.}
\end{align*}
By the same reasoning as before we can infer that $|N_{W_{\slav}}| \geq \frac{3|N_{W_{\cc}}|}{5}$.

For the upper bound we will use the following construction (similar as in the proof for PAV).
The set of candidates is $X\cup Y$ with $X = \{x_1, \ldots, x_k\}$ and $Y = \{y_1, \ldots, y_{k-1}\}$. Let $z$ be an integer divisible by $2k-1$. 
There are $z$ voters who approve $X$. Further, for each $i\in[k-1]$, there are $\frac{z}{2k-1}$ voters who approve candidate $y_i$. 
Thus:
\begin{align*}
n = z + (k-1) \cdot \frac{z}{2k-1} = z \cdot \frac{3k-2}{2k-1}\text{.}
\end{align*}
Committee $X$ is winning according to SLAV (tied with other committees)---and $X$ has a CC-score of $z$. The optimal CC committee consists of a single candidate from $X$ and all $(k-1)$ candidates from $Y$---this would give a CC-score of $n$. Thus, the representation guarantee of SLAV is at most equal to:
\begin{align*}
\frac{z}{n} = \frac{2k-1}{3k-2}=\frac{2}{3}+\frac{1}{9k-6}\text{.}
\end{align*}
\end{proof}

\begin{reptheorem}{thm:seq_pav_guarantee-combined} 
\thmseqpavguaranteecombined
\end{reptheorem}

\begin{proof}[Proof of the utilitarian guarantee]
Since seq-PAV satisfies lower quota~\citep{BLS17a},\footnote{An ABC rule $\calR$ satisfies \emph{lower quota} if for every $k \in \naturals$, every party-list profile $A$, and party $N'$ it holds that every winning committee from $\calR(A, k)$ contains at least $\left\lfloor\frac{k\cdot |N'|}{|N|}\right\rfloor$ candidates that are approved by $N'$.} it also satisfies weak proportionality. Hence the upper bound of $\frac{2}{\lfloor \sqrt{k} \rfloor} - \frac{1}{k}$ follows from \Cref{prop:AvGuaranteeOfCC}. In the remaining part of the proof we will prove the lower bound.

For $k=1$, seq-PAV is AV and hence the utilitarian guarantee is 1. For $k=2$, in the first step the AV-winner is chosen and hence we have a utilitarian guarantee for $k=2$ is $\frac 3 4\geq \frac{1}{2\sqrt{2}}$.
Now assume that $k \geq 3$.
Let $W_{\pav}^{(j)}$ denote the first $j$ candidates selected by sequential PAV; in particular, $W_{\pav}^{(0)} = \emptyset$. Let $w_j$ denote the candidate selected by sequential PAV in the $j$th step, thus $w_j$ is the single candidate in the set $W_{\pav}^{(j)} \setminus W_{\pav}^{(j-1)}$.
Let $x_{i, j} = |W_{\pav}^{(j)} \cap A(i)|$. Next, let $W_{\av}$ be the optimal committee according to Approval Voting, and let $s_{\av} = \score_{\av}(A,W_{\av})$.

If at some step $j$ of the run of sequential PAV, it happens that the AV-score of $W_{\pav}^{(j)}$, which is $\sum_{i \in N}x_{i, j}$, is greater or equal than $\frac{s_{\av}}{2\sqrt{k}}$, then our hypothesis is clearly satisfied. Thus, from now on, we assume that for each $j$ we have that $\sum_{i \in N}x_{i, j} < \frac{s_{\av}}{2\sqrt{k}}$. Also, this means that in each step there exists a candidate $c$ from $W_{\av} \setminus W_{\pav}$ who is approved by $n_c \geq \frac{s_{\av} - \frac{s_{\av}}{2\sqrt{k}}}{k} \geq \frac{s_{\av}}{k}(1 - \frac{1}{2\sqrt{3}})$ voters. Let $n_c = |N(c)|$.

Let $\Delta p_j$ denote the increase of the PAV-score due to adding $w_{j+1}$ to $W_{\pav}^{(j)}$. %
Using the inequality between harmonic and arithmetic mean, 
we have that:
\begin{align*}
\Delta p_j &= \sum_{i \in N(c)} \frac{1}{x_{i, j} + 1} \geq \frac{n_c^2}{\sum_{i \in N(c)} x_{i, j} + n_c} > \frac{n_c^2}{\frac{s_{\av}}{2\sqrt{k}} + n_c} \\
           &\geq \frac{\left(\frac{s_{\av}}{k}(1 - \frac{1}{2\sqrt{3}})\right)^2}{\frac{s_{\av}}{2\sqrt{k}} + \frac{s_{\av}}{k}(1 - \frac{1}{2\sqrt{3}})} \geq \frac{\left(\frac{s_{\av}}{k}(1 - \frac{1}{2\sqrt{3}})\right)^2}{\frac{s_{\av}}{\sqrt{k}}(\frac{1}{2} + \frac{1}{\sqrt{3}} - \frac{1}{6})} \\
           &= \frac{s_{\av}}{k\sqrt{k}} \cdot  \underbrace{\frac{\left(1 - \frac{1}{2\sqrt{3}}\right)^2}{\frac{1}{2} + \frac{1}{\sqrt{3}} - \frac{1}{6}}}_{\approx 0.56} >  \frac{s_{\av}}{2k\sqrt{k}} \text{.}
\end{align*}  
Since this must hold in each step of sequential PAV, we get that the total PAV-score of $W_{\pav}^{(k)}$ must be at least equal to $k \cdot \frac{s_{\av}}{2k\sqrt{k}} = \frac{s_{\av}}{2\sqrt{k}}$. Since the AV-score is at least equal to the PAV-score of any committee, we obtain a contradiction and conclude that $\score_{\av}(A,W_{\pav}^{(k)}) \geq \frac{s_{\av}}{2\sqrt{k}}$.
\end{proof} 

\begin{proof}[Proof of the representation guarantee lower bound]
Consider an approval profile $A$ and let $W_{\spav}$ and $W_{\cc}$ denote the winning committees in~$A$ according to seq-PAV and CC, respectively. Let $n_{\spav} = \score_{\cc}(A, W_{\spav})$ and $n_{\cc} = \score_{\cc}(A, W_{\cc})$. The total PAV-score of $W_{\spav}$ is at most equal to $n_{\spav}\H(k) \leq n_{\spav}(\log(k) + 1)$. Thus, at some step sequential PAV selected a committee member who improved the PAV-score by at most $\frac{n_{\spav}(\log(k) + 1)}{k}$. On the other hand, by the pigeonhole principle, we know that at each step of seq-PAV there exists a not-selected candidate whose selection would improve the PAV-score by at least $\frac{n_{\cc} - n_{\spav}}{k}$. Consequently, we get that 
\begin{align*}
\frac{n_{\spav}(\log(k) + 1)}{k} \geq \frac{n_{\cc} - n_{\spav}}{k} \text{.}
\end{align*} 
After reformulation we have that $n_{\spav} \geq \frac{n_{\cc}}{\log(k) + 2}$, which completes the proof.
\end{proof}

\begin{reptheorem}{thm:pgeom-combined}\thmpgeom
\end{reptheorem}

\begin{proof}[Proof of the utilitarian guarantee]
We use the same notation as in the proof of \Cref{thm:guarantees_for_pav} with the difference that instead of $W_{\pav}$ (denoting a PAV winning committee) we will use $W_{\pgeom}$, denoting a committee winning according to the $p$-Geometric rule. 
For each two candidates, $c \in W_{\pgeom}$ and $c'$ arbitrary, let $\Delta_p(c', c)$ denote the change of the $p$-Geometric score of $W_{\pgeom}$ caused by replacing $c$ with $c'$. 
We obtain:
\begin{align*}
\sum_{c' \in W_{\pgeom}}\Delta_p(c, c') &= k \sum_{i \in N(c)} \left(\frac{1}{p}\right)^{w_i+1} - \sum_{i \in N} \sum_{c' \in W_{\pav} \cap A(i)} \left(\frac{1}{p}\right)^{w_i} \\
                                    &= k \sum_{i \in N(c)} \left(\frac{1}{p}\right)^{w_i+1} - \sum_{i \in N} w_i \left(\frac{1}{p}\right)^{w_i} 
\end{align*}
By using Jensen's inequality we get that $\sum_{i \in N(c)} \frac{1}{n_c} \cdot \left(\frac{1}{p}\right)^{w_i+1} \geq \left(\frac{1}{p}\right)^{\frac{\sum_{i \in N(c)}w_i + n_c}{n_c}}$. Thus:
\begin{align*}
\sum_{c' \in W_{\pgeom}}\Delta_p(c, c') &= k n_c \left(\frac{1}{p}\right)^{\frac{\sum_{i \in N}w_i}{n_c} + 1} - \sum_{i \in N} w_i \left(\frac{1}{p}\right)^{w_i} \\
                                      &\geq k n_c \left(\frac{1}{p}\right)^{\frac{\sum_{i \in N}w_i}{n_c} + 1} - \frac{1}{p}\sum_{i \in N} w_i
\end{align*}
Since we know that $\sum_{c' \in W_{\pgeom}}\Delta_p(c, c') \leq 0$, we have that:
\begin{align*}
\frac{1}{p}\sum_{i \in N} w_i \geq k n_c \left(\frac{1}{p}\right)^{\frac{\sum_{i \in N}w_i}{n_c} + 1}
\end{align*}
Let us set $r = \frac{kn_{c}}{\sum_{i \in N} w_i}$, and observe (similarly as in the proof of \Cref{thm:guarantees_for_pav}) that \[\frac{\score_{\av}(A,W_{\av})}{\score_{\av}(A,W_{\pgeom})} \leq 1 + r.\] We have that $p^{\frac{k}{r}} \geq r$. The equation $p^{\frac{k}{r}} = r$ has only one solution, $r= \frac{k \log(p)}{\W(k \log(p))}$. This gives $r \leq \frac{k \log(p)}{\W(k \log(p))}$ and proves that the utilitarian guarantee is greater than or equal to 
\begin{align*}
\frac{\W(k \log(p))} {k \log(p) + \W(k \log(p))}.
\end{align*}

Now, let us prove the upper bound on the utilitarian guarantee.
Let $z = \frac{k \log(p)}{\W(k \log(p))}$; in particular, by the properties of the Lambert function we have that $z = p^{\frac{k}{z}}$.
Consider the following instance. Let $x$ be a large integer so that $\left\lfloor x \cdot z \right\rfloor \approx xz$. (Formally, we choose an increasing sequence $\bar x$ so that $z\bar x - \left\lfloor z\bar x \right\rfloor\to 0$.)
Assume there are $\left\lfloor x \cdot z \right\rfloor$ voters who approve candidates $B = \{c_1, \ldots, c_k\}$. Additionally, for each candidate $c \in D = \{c_{k+1}, \ldots, c_{2k}\}$ there are $x$ distinct voters who approve $c$. For this instance the $p$-Geometric rule selects at most $\left\lceil \frac{k}{z} \right\rceil$ members from $B$: if more candidates from $B$ were selected, then replacing one candidate from $B$ with a candidate from $D$ would increase the $p$-Geometric-score by more than \[\frac{x}{p} - \left\lfloor x \cdot z\right\rfloor \cdot \left(\frac{1}{p}\right)^{\left\lceil\frac{k}{z}\right\rceil+1}> \frac{x}{p} - \frac{x}{p} \cdot z \cdot \left(\frac{1}{p}\right)^{\frac{k}{z}} =  \frac{x}{p} - \frac{x}{p} \cdot z \cdot \left(\frac{1}{z}\right) = 0,\] a contradiction. Thus, the AV-score of the committee selected by the $p$-Geometric rule would be smaller than $x \cdot z \cdot \left(1 + \frac{k}{z}\right) + kx = xz + 2kx$. Thus, we get that the utilitarian guarantee of the $p$-Geometric rule is at most equal to:
\begin{align*}
\frac{2kx + xz}{xzk} = \frac{1}{k} + \frac{2}{z} = \frac{1}{k} + \frac{2\W(k \log(p))}{k \log(p)} \text{.}
\end{align*}
\end{proof}

\begin{proof}[Proof of the representation guarantee]
Let $A$ be an approval profile and let $W_{\cc}$ and $W_{\pgeom}$ be two committees winning according to the Chamberlin--Courant and $p$-Geometric rule, respectively. Let $n_{\pgeom} = \score_{\cc}(A, W_{\pgeom})$ and $n_{\cc} = \score_{\cc}(A, W_{\cc})$.
We observe that:
\begin{align*}
\score_{\pgeom}(A, W_{\pgeom}) \leq n_{\pgeom}\left(\frac{1}{p} + \frac{1}{p^2} + \ldots\right) \leq n_{\pgeom}\cdot \frac{1}{p} \cdot \frac{1}{1 - \frac{1}{p}}
\end{align*}
and that:
\begin{align*}
\score_{\pgeom}(A, W_{\cc}) \geq n_{\cc} \cdot \frac{1}{p} \text{.}
\end{align*}
Consequently, from $ \score_{\pgeom}(A, W_{\pgeom}) \geq \score_{\pgeom}(A, W_{\cc})$ we get that:
\begin{align*}
n_{\pgeom}\cdot \frac{1}{1 - \frac{1}{p}} \geq p\cdot \score_{\pgeom}(A, W_{\pgeom}) \geq p\cdot \score_{\pgeom}(A, W_{\cc}) \geq n_{\cc} \text{,}
\end{align*}
which gives the lower bound on the representation guarantee.

Now, let us prove the upper bound. Fix a rational number $p$ and some large integer $x$ such that $px$ is integer.
First, let $k$ be even with $k=2k'$.
Let the set of candidates be $\{x_1,\dots,x_k\}\cup\{y_1,\dots,y_{k'}\}$.
There are $k'$ groups of voters who consists of $px$ voters; in each group voters approve some two distinct candidates from $\{x_1,\dots,x_k\}$. Additionally, there are $k'$ groups consisting of $x$ voters who approve some distinct candidate from $\{y_1,\dots,y_{k'}\}$. It is easy to see that for such instances the representation guarantee is at most equal to $\frac{k'px}{k'px + k'x} = \frac{p}{1+p}$.

Now, let $k$ be odd with $k=2k'+1$; the set of candidates is $\{x_1,\dots,x_{2k'+2}\}\cup\{y_1,\dots,y_{k'}\}$.
There are $k'+1$ groups of voters who consists of $px$ voters; in each group voters approve some two distinct candidates from $\{x_1,\dots,x_{2k'+2}\}$. Additionally, there are $k'$ groups consisting of $x$ voters who approve some distinct candidate from $\{y_1,\dots,y_{k'}\}$.
Now, we see that the for such instances the representation guarantee is at most equal to \[\frac{(k'+1)px}{(k'+1)px + k'x} = \frac{p}{p+1-\frac{1}{k'+1}}=\frac{p}{p+1-\frac{2}{k+2}}=\frac{p}{p+\frac{k}{k+2}}.\]
The upper bound for the odd case is larger and hence prevails.
\end{proof}

\begin{reptheorem}{thm:phragmen_bounds-combined}
\thmphrag
\end{reptheorem}

\begin{proof}[Proof of the utilitarian guarantee]
First, we will prove the lower bound of $\frac{1}{5\sqrt{k} + 1}$.
Consider an approval profile $A$, and let $W_{\phrag}$ and $W_{\av}$ be committees winning according to seq-Phragm\'{e}n and AV, respectively. W.l.o.g., we assume that $W_{\phrag} \neq W_{\av}$. For iteration $t$ we will use the following notation:
\begin{enumerate}[(1)]
\item Let $w^{(t)}_{\phrag}$ be the candidate selected by seq-Phragm\'{e}n in the $t$-th iteration. Further, let $w^{(t)}_{\av}$ be a candidate with the highest AV-score in $W_\av\setminus \{w^{(1)}_{\phrag},\dots, w^{(t-1)}_{\phrag}\}$.
\item Let $n^{(t)}_{\phrag} = |N(w^{(t)}_{\phrag})|$, and $n^{(t)}_{\av} = |N(w^{(t)}_{\av})|$.
\item Let $\ell_j(t)$ denote the total load assigned to voter $j$ until $t$. The \emph{maximum load} in iteration $t$ is $\max_{j\in N} \ell_j(t)$.
\item Let $\ell^{(t)}_{\av}$ denote the total load distributed to the voters from $N(w^{(t)}_{\av})$ until iteration $t$, and let $m^{(t)}_{\av}$ denote the maximum load assigned to a voter from $N(w^{(t)}_{\av})$ until~$t$, i.e., $m^{(t)}_{\av} = \max_{j\in N(w^{(t)}_{\av})} \ell_j(t)$.
\end{enumerate}

We will use an argument based on a potential function $\Phi \colon [0,t]\to \reals$, which we maintain during each iteration of seq-Phragm\'{e}n. Let $\Phi(0)=0$. In iteration $t$, we increase the potential function by $\big(5\sqrt{k} + 1\big) \cdot n^{(t)}_{\phrag}$ and decrease it by $n^{(t)}_{\av}$, i.e., \[\Phi(t) = \Phi(t-1) +  \big(5\sqrt{k} + 1\big) \cdot n^{(t)}_{\phrag} - n^{(t)}_{\av}.\] Our goal is to show that $\Phi(k)\geq 0$.
If we know that $\Phi(k)>0$, we can infer that
\[\sum_{t=1}^k\big(5\sqrt{k} + 1\big) \cdot n^{(t)}_{\phrag}- \sum_{c\in W_\av} |N(c)|\geq \sum_{t=1}^k\big(5\sqrt{k} + 1\big) \cdot n^{(t)}_{\phrag}- \sum_{t=1}^k n^{(t)}_{\av} = \Phi(k) \geq 0.\]
and hence
the utilitarian guarantee of seq-Phragm\'{e}n is lower-bounded by $\frac{1}{5\sqrt{k} + 1}$.

Let $s$ be the first iteration where $\ell^{(s)}_{\av} > 3\sqrt{k}$; if $\ell^{(t)}_{\av} \leq 3\sqrt{k}$ for all $t\in[k]$ then we set $s=k+1$. 

First, let us consider iterations $t<s$ and show that $\Phi(t)\geq \Phi(t-1)+n^{(t)}_{\phrag} \cdot 2\sqrt{k}$.
If $w^{(t)}_{\phrag} = w^{(t)}_{\av}$, then $\Phi(t) = \Phi(t-1) +  \big(5\sqrt{k} \big) \cdot n^{(t)}_{\phrag} $.
Let us assume $w^{(t)}_{\phrag} \neq w^{(t)}_{\av}$.
We first show that $m^{(t)}_{\av} \leq \frac{\ell^{(t)}_{\av} + 1}{n^{(t)}_{\av}}$. For the sake of contradiction assume 
that $t$ is the first iteration after which 
$m^{(t)}_{\av}>\frac{\ell^{(t)}_{\av} + 1}{n^{(t)}_{\av}}$.
First note that this is only possible if indeed $w^{(t)}_{\av} \neq w^{(t)}_\phrag$.
However, by selecting $w^{(t)}_{\av}$ instead of $w^{(t)}_{\phrag}$, it can be ensured that the load does not increase above~$\frac{\ell^{(t)}_{\av} + 1}{n^{(t)}_{\av}}$, so seq-Phragm\'{e}n would have chosen $w^{(t)}_{\av}$, a contradiction. Next, observe that after $w^{(t)}_{\phrag}$ has been selected, the largest load assigned in total to a voter is at least equal to $\nicefrac{1}{n^{(t)}_{\phrag}}$. Yet, if $w^{(t)}_{\av}$ were selected, then the largest total load assigned to a voter would be at most equal to $\frac{\ell^{(t)}_{\av} + 1}{n^{(t)}_{\av}}$. Thus, it must hold that $\frac{\ell^{(t)}_{\av} + 1}{n^{(t)}_{\av}} \geq \nicefrac{1}{n^{(t)}_{\phrag}}$, which is equivalent to $n^{(t)}_{\av} \leq n^{(t)}_{\phrag}(\ell^{(t)}_{\av} + 1)$. 
It follows that $n^{(t)}_{\av} \leq n^{(t)}_{\phrag}(3\sqrt{k} + 1)$. Consequently, we have that 
\begin{align}
\Phi(t) &\geq \Phi(t-1) + \left(5\sqrt{k} + 1\right) \cdot n^{(t)}_{\phrag} - n^{(t)}_{\av} \\&\geq \left(5\sqrt{k} + 1\right) \cdot n^{(t)}_{\phrag} - \left(3\sqrt{k} + 1\right) \cdot n^{(t)}_{\phrag} = n^{(t)}_{\phrag} \cdot 2\sqrt{k} \text{.}\label{eq:pot-funct}
\end{align}

Now, we bound $\Phi(s-1)$.
Let $w=w^{(s+1)}_{\av}$, i.e., let $w$ be a candidate with the highest AV-score contained in $W_\av\setminus \{w^{(1)}_{\phrag},\dots, w^{(s)}_{\phrag}\}$; let $n_w=|N(w)|$.
Here, we divide our reasoning into the following sequence of claims:
\begin{enumerate}[(1)]
\item Observe that in step $s$, a candidate other than $w$ is selected by seq-Phragm\'{e}n and selecting candidate $w$ would increase the maximum load by at most $\nicefrac{1}{n_w}$.
As a consequence, in each iteration $t\leq s$, the maximum load increased by at most $\nicefrac{1}{n_w}$.

\item We will show that the following holds: if the maximum load in $N(w)$ increases by at least $\nicefrac{2}{n_w}$ between two iterations $t_1$ and $t_2 \leq s$, then the AV-score from voters in $N(w)$ increased between these two iterations by at least $\frac{n_w}{2}$. 
Towards a contradiction, assume that this is not the case, i.e., that between $t_1$  and $t_2$ the maximum load from voter in $N(w)$ increases by at least $\nicefrac{2}{n_w}$, and the load of more than $\nicefrac{n_w}{2}$ voters in $N(w)$ does not increase. Without loss of generality, assume that $t_2$ is the first iteration for which our assumption holds. Then, if in $t_2$ we selected $w$ and distributed its load among these more than $\nicefrac{n_w}{2}$ voters whose load has not yet increased, then the maximum load would increase by less than $\nicefrac{2}{n_w}$. This contradicts the fact that seq-Phragm\'{e}n does not choose $w$ (by definition of $w$).

\item Let us group the iterations of seq-Phragm\'{e}n before $s$ into blocks. The $i$-th block starts after the $(i-1)$-th block ends (the first block starts with the first iteration). Further, each block ends right after the first iteration which increases the maximum load assigned to a voter from $N(w)$ by at least $\nicefrac{2}{n_w}$ since the moment the block has started (thus, the last iterations may not be part of a block). Thus, in each block the maximum load assigned to a voter from $N(w)$ increases by at least $\nicefrac{2}{n_w}$.
Since in one step the load can increase by no more than $\nicefrac{1}{n_w}$, in each block the maximum load assigned to a voter from $N(w)$ increases by at most $\nicefrac{2}{n_w} + \nicefrac{1}{n_w} = \nicefrac{3}{n_w}$. Consequently, since $\ell^{(s)}_{\av} > 3\sqrt{k}$ (and so, by the pigeonhole principle, some voter from $N(w)$ is assigned the load at least equal to $\frac{3\sqrt{k}}{n_w}$), until $s$ there are at least $\sqrt{k}$ blocks. By the previous point, the total AV-score of voters increases in each block by at least $\nicefrac{n_w}{2}$.
Since there are at least $\sqrt{k}$ blocks, we have that
\begin{align*}
\sum_{t=1}^{s-1}n^{(t)}_{\phrag}\geq \sqrt{k}\cdot \nicefrac{n_w}{2}.
\end{align*}
By Equation~\eqref{eq:pot-funct}, we have that
\begin{align*}
\Phi(s-1) \geq \sqrt{k}\cdot \nicefrac{n_w}{2} \cdot 2\sqrt{k} = k n_w \text{.}
\end{align*}
\end{enumerate}
By choice of $w$, candidates not contained in $W_{\phrag}$ are approved by at most $n_w$ voters and hence
$\Phi(k) - \Phi(s-1) \geq -kn_w \text{.}$
Hence $\Phi(k) \geq 0$.
This concludes the lower bound proof.

For the upper bound we use the fact that seq-Phragm\'{e}n satisfies the lower quota property~\citep{BLS17a} and thus also weak proportionality; hence the upper bound from  \Cref{thm:proportional_guarantees} applies.
This completes the proof.
\end{proof}

\begin{proof}[Proof of the representation guarantee]
We first prove the lower bound on the representation guarantee of seq-Phragm\'{e}n. Consider an approval profile $A$, and let $W_{\phrag}$ be a committee selected by seq-Phragm\'{e}n for $A$; let $W_{\cc}$ be a committee maximizing the CC-score for $A$. Further, for each $i$, $1 \leq i \leq k$, by $W^{(i)}_{\phrag}$ we denote the first $i$ candidates selected by seq-Phragm\'{e}n. We set $n_{\cc} = |N(W_{\cc})|$ and $n^{(i)}_{\phrag}= |N(W^{(i)}_{\phrag})|$.

We will show by induction that for each $i$ it holds that $n^{(i)}_{\phrag} \geq \frac{i \cdot n_{\cc}}{k+i}$. For $i = 0$, the base step of the induction is trivially satisfied. Now, assume that for some $i$ we have $n^{(i)}_{\phrag} \geq \frac{i \cdot n_{\cc}}{k+i}$, and we consider the $(i+1)$-th step of seq-Phragm\'{e}n. Observe that there exists a not-yet selected candidate $c$ who is supported by at least $\frac{n_{\cc}-n^{(i)}_{\phrag}}{k}$ voters who do not have yet a representative in $W^{(i)}_{\phrag}$. Consider the following two cases:
\begin{description}
\item[Case 1:] $c$ is not selected in the $(i+1)$-th step. After this step the maximum load assigned to a voter is at least equal to $\frac{i+1}{n^{(i+1)}_{\phrag}}$, which is the number of chosen candidates divided by the number of voters that share their load. By selecting $c$ the load would increase to no more than $\frac{k}{n_{\cc}-n^{(i+1)}_{\phrag}}$. Consequently, we have that $\frac{k}{n_{\cc}-n^{(i+1)}_{\phrag}} \geq \frac{i+1}{n^{(i+1)}_{\phrag}}$. This is equivalent to $n^{(i+1)}_{\phrag} \geq \frac{(i+1) n_{\cc}}{k + i + 1}$.
\item[Case 2:] $c$ is selected in the $(i+1)$-th step. Then, $n^{(i+1)}_{\phrag} \geq n^{(i)}_{\phrag} + \frac{n_{\cc}-n^{(i)}_{\phrag}}{k}$. After reformulating:
\begin{align*}
n_{\cc} - n^{(i+1)}_{\phrag} \leq n_{\cc} - n^{(i)}_{\phrag} - \frac{n_{\cc}-n^{(i)}_{\phrag}}{k} = (n_{\cc} - n^{(i)}_{\phrag}) \cdot \frac{k-1}{k} \text{.}
\end{align*}
By the inductive assumption we have $n_{\cc} - n^{(i)}_{\phrag} \leq n_{\cc} - \frac{n_{\cc} i}{k+i} = \frac{n_{\cc} k}{k+i}$ and
\begin{align*}
n_{\cc} - n^{(i+1)}_{\phrag} \leq \frac{n_{\cc} k}{k+i} \cdot \frac{k-1}{k} = \frac{n_{\cc} (k-1)}{k+i} \text{.}
\end{align*}
Consequently,
\begin{align*}
n^{(i+1)}_{\phrag} \geq n_{\cc} - \frac{n_{\cc} (k-1)}{k+i} = \frac{n_{\cc} (i+1)}{k+i} \geq \frac{n_{\cc}(i+1)}{k+i+1} \text{.}
\end{align*}
\end{description}
In both cases the inductive step is satisfied, which shows that our hypothesis holds. In particular, for $i = k$, we have that $n^{(k)}_{\phrag} \geq \frac{kn_{\cc}}{k+k} = \frac{n_{\cc}}{2}$. This proves the lower bound on the representation guarantee of seq-Phragm\'{e}n.

For the upper bound we use the same construction and argument as in the proof of the representation guarantee of PAV (\Cref{thm:guarantees_for_pav}).
\end{proof}

\begin{reptheorem}{thm:phragmen_opt_bounds-combined}
\thmoptphrag
\end{reptheorem}

\begin{proof}[Proof of the utilitarian guarantee]
We first prove the lower-bound. Let $W$ be a committee selected by leximin-\phragmen{} and let $c$ be a candidate approved by most voters. For the sake of contradiction assume that the approval score of $W$ is lower than $|N(c)|$, which in particular means that $c \notin W$. Take a candidate $c' \in W$, and assume there are $x$ voters who approve $c'$ and do not approve $c$, $|N(c') \setminus N(c)| = x$. Thus, since $|N(c) > N(W)|$, there are more than $x$ voters who approve $c$ and none of the candidates from $W$. Now, consider a committee obtained from $W$ by replacing $c'$ with $c$, and the load distribution constructed from the distribution for $W$ as follows. The voters who approve both $c$ and $c'$ get from $c$ the same load as they did from $c'$. The same amount of load that was distributed from $c'$ to the $x$ voters from $N(c') \setminus N(c)$ is now distributed evenly among the voters from $N(c) \setminus N(W)$. This gives a load distribution that is more preferred according to leximin-\phragmen{} and leads to the contradiction with the optimality of~$W$.

For the upper bound consider a profile with $m = 2k$ candidates, where $n$ voters are divided into $k$ equal-size groups (we assume that $n$ is divisible by $k$): $N = N_1 \cup N_2 \cup \ldots \cup N_k$. For $i \in [k]$ candidate $c_i$ is approved by group $N_i$. The remaining $k$ candidates ($\{c_i \colon k+1 \leq i \leq 2k\}$) are identical and each of them is approved by the first $\nicefrac{n}{k}-1$ voters from $N_1$, the first $\nicefrac{n}{k}-1$ voters from $N_2$, etc. For this profile leximin-\phragmen{} selects committee $\{c_1, \ldots, c_k\}$, which in total collects $n$ approvals. On the other hand, committee $\{c_{k+1}, \ldots, c_{2k}\}$ gets $k(n-k)$ approvals. Since $n$ can be arbitrarily large, we get the ratio of $\nicefrac{1}{k}$. 
\end{proof}

\begin{proof}[Proof of the representation guarantee]
We start by proving the lower bound. Let $W$ be a committee returned by leximin-\phragmen{}, and let $W_{\cc}$ be a committee maximizing the CC-score. Let $n_{W} = |N(W)|$, $n_{\cc} = |N(W_{\cc})|$, and for the sake of contradiction assume that $2n_{W} < n_{\cc}$. Observe that the maximal load assigned to a voter is higher than $\nicefrac{2k}{n_{\cc}}$. Let $c$ be a candidate who assigns some positive amount of load to a maximally loaded voter. On the other hand, there exists a candidate $c' \in W_{\cc} \setminus W$ who is approved by at least $\nicefrac{n_{\cc}}{2k}$ voters from $N(W_{\cc}) \setminus N(W)$. By replacing $c$ with $c'$ in $W$ we can remove the load assigned from $c$ (decreasing the load of at least one voter who initially had load higher than $\nicefrac{2k}{n_{\cc}}$) and instead spread the load from $c'$ to the voters from $N(W_{\cc}) \setminus N(W)$---the total load of these voters will not exceed $\nicefrac{2k}{n_{\cc}}$. Thus, the new load distribution is preferred to the old one by leximin-\phragmen{}, a contradiction.  

For the upper bound we use the same construction and argument as in the proof of the representation guarantee of PAV (\Cref{thm:guarantees_for_pav}).
\end{proof}

\begin{reptheorem}{thm:mnimax_bounds-combined}
\thmminimax
\end{reptheorem} 
\begin{proof}
 Let us fix $k$, and consider the following instance with $4k$ candidates. The first voter's approval set is $A(1) = \{c_1, \ldots, c_{3k}\}$, the remaining $n-1$ voters approve $A(2) = \ldots = A(n) = \{c_{3k +1}, \ldots, c_{4k}\}$. For this profile MAV selects the committee consisting of $k$ candidates from~$A(1)$. Indeed, such committees result in the maximum Hamming distance of $2k$; selecting less candidates from~$A(1)$ would result in the Hamming distance equal to at least $2k + 2$.

This example shows that the utilitarian guarantee cannot be higher than $\frac{k}{(n-1)k} = \frac{1}{(n-1)}$, and that the representation guarantee cannot be higher than $\frac{1}{n}$. Since the number of voters $n$ can be arbitrarily large, we get that the two guarantees equal zero. 
\end{proof}

\section{Further Experimental Details from \cref{sec:average_guarantees}}\label{app:exp}

This section contains the plots and numerical data for the Mallows and urn datasets. 
\Cref{tab:util2,tab:egal2} show the respective boxplots; \Cref{fig:mallows-2,fig:urn-2} show the average utilitarian and representation ratios.

\begin{figure}
\includegraphics[width=\textwidth]{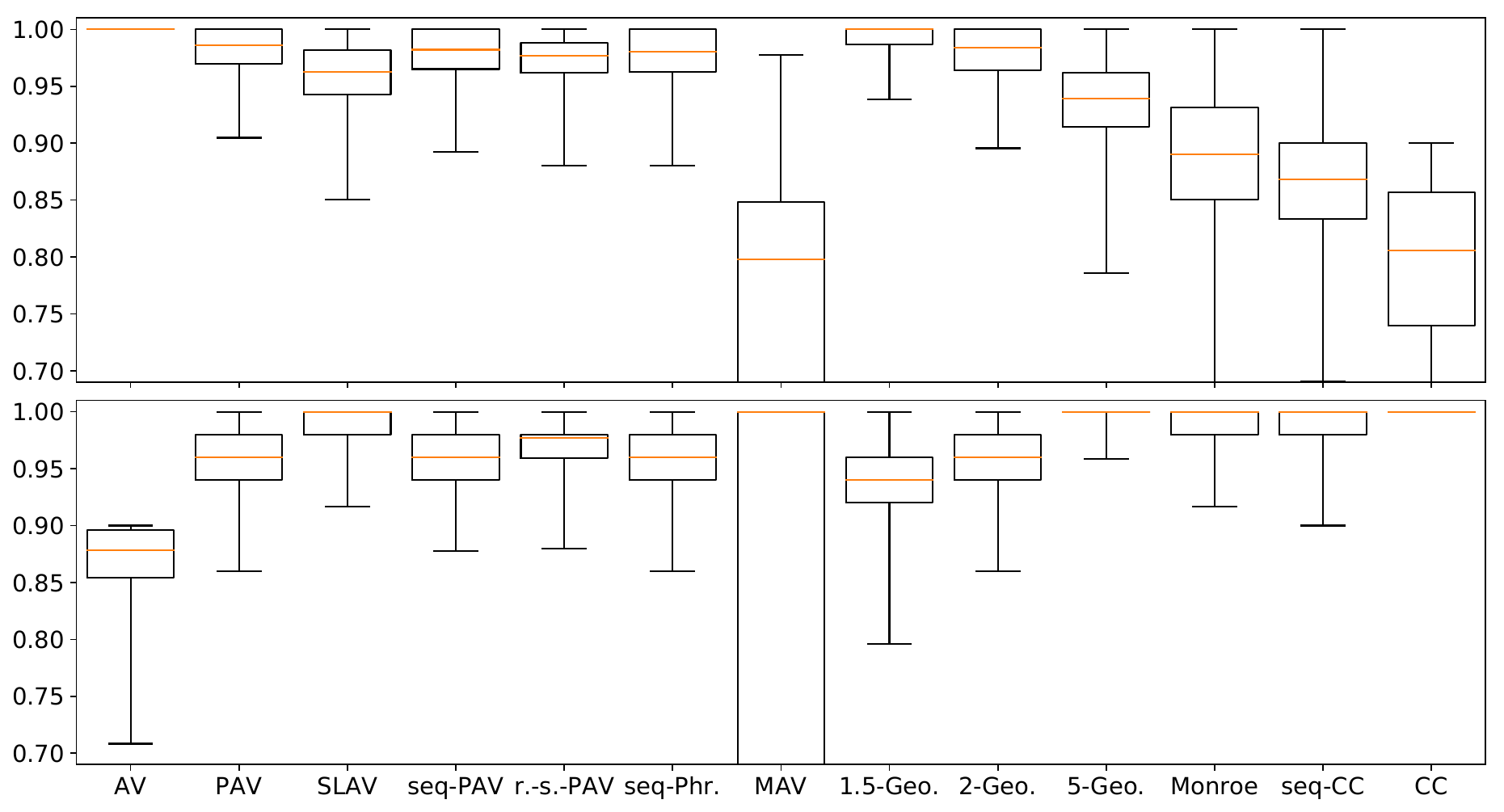}%
\caption{Results for the Mallows dataset (upper boxplot shows utilitarian ratios, the lower representation ratios).}
\label{fig:mallows-2}
\end{figure}

\begin{figure}
\includegraphics[width=\textwidth]{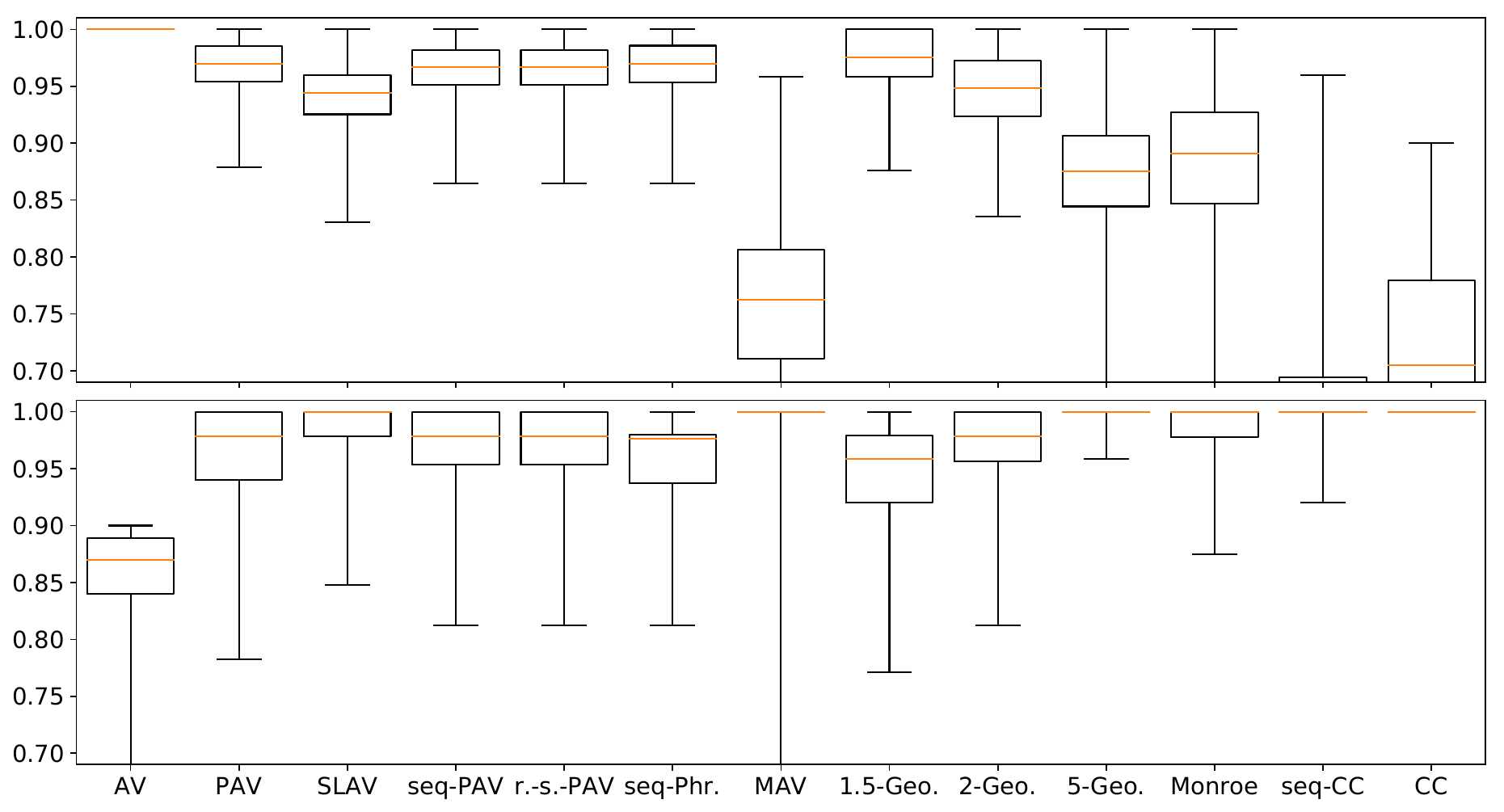}%
\caption{Results for the urn dataset (upper boxplot shows utilitarian ratios, the lower representation ratios).}
\label{fig:urn-2}
\end{figure}

\begin{table}
\begin{center}
\subcaptionbox{Mallows dataset}{
\begin{tabular}{clrc}
    &                & mean  & \\
    \toprule
 1. &AV              &1.000& \\
 2. &$1.5$-Geometric &0.993& \\
 3. &PAV             &0.982& \\
 4. &$2$-Geometric   &0.979& \\
 5. &seq-PAV         &0.979& \\
 6. &seq-\phragmen   &0.978& \\
 7. &rev-seq-PAV     &0.976& \\
 8. &SLAV            &0.960& \\
 9. &$5$-Geometric   &0.937& \\ 
10. &Monroe          &0.888& \\
11. &seq-CC          &0.866& \\
12. &CC              &0.793& \\
13. &MAV             &0.535& \\
\end{tabular}}
\qquad\qquad
\subcaptionbox{urn dataset}{
\begin{tabular}{clrc}
    &                & mean  & \\
    \toprule
 1. &AV              &1.000& \\
 2. &$1.5$-Geometric &0.974& \\
 3. &PAV             &0.968& \\
 4. &seq-\phragmen   &0.968& \\
 5. &rev-seq-PAV     &0.965& \rdelim{]}{2}{1mm}[]\\
 6. &seq-PAV         &0.965& \\
 7. &$2$-Geometric   &0.947& \\
 8. &SLAV            &0.942& \\
 9. &Monroe          &0.883& \\
10. &$5$-Geometric   &0.875& \\ 
11. &MAV             &0.723& \\
12. &CC              &0.644& \\
13. &seq-CC          &0.602& \\
\end{tabular}}%
\end{center}%
\caption{Mean utilitarian ratios for the Mallows and urn dataset. The differences between pairs of rules are statistically significant (paired t-test, $p=0.01$), unless a pair is marked with a bracket.}\
\label{tab:util2}
\end{table}

\begin{table}
\begin{center}
\subcaptionbox{Mallow dataset}{
\begin{tabular}{clrc}
    &                & mean  & \\
    \toprule
 1. &CC             & 1.000& \\
 2. &$5$-Geometric  & 0.999& \\
 3. &Monroe         & 0.991& \rdelim{]}{2}{1mm}[]\\
 4. &seq-CC         & 0.990& \\
 5. &SLAV           & 0.988& \\
 6. &rev-seq-PAV    & 0.969& \\
 7. &seq-PAV        & 0.964& \rdelim{]}{2}{1mm}[]\\
 8. &$2$-Geometric  & 0.964& \\
 9. &PAV            & 0.963& \\
10. &seq-\phragmen  & 0.960& \\
11. &$1.5$-Geometric& 0.944& \\
12. &AV             & 0.868& \\
13. &MAV            & 0.633& \\
\end{tabular}}
\qquad\qquad
\subcaptionbox{urn dataset}{
\begin{tabular}{clrc}
    &                & mean  & \\
    \toprule
 1. &CC             & 1.000& \\
 2. &$5$-Geometric  & 1.000& \\
 3. &seq-CC         & 0.999& \\
 4. &SLAV           & 0.985& \\
 5. &Monroe         & 0.983& \\
 6. &$2$-Geometric  & 0.971& \\
 7. &rev-seq-PAV    & 0.967& \rdelim{]}{2}{1mm}[]\\
 8. &seq-PAV        & 0.967& \\
 9. &PAV            & 0.965& \\
10. &seq-\phragmen  & 0.960& \\
11. &$1.5$-Geometric& 0.951& \rdelim{]}{2}{1mm}[]\\
12. &MAV            & 0.946& \\
13. &AV             & 0.860& \\
\end{tabular}}%
\end{center}%
\caption{Mean representation ratios for the Mallows and urn dataset. The differences between pairs of rules are statistically significant (paired t-test, $p=0.01$), unless a pair is marked with a bracket.}
\label{tab:egal2}
\end{table}	

\newpage
\section{Proof Details from \Cref{sec:efficiency_axiom}}\label{app:axiom}

\begin{example}
Consider the following profile with 20 voters and 4 candidates, where:
\begin{align*}
& N(c_1) = \{2, \ldots, 10\}; \quad && N(c_2) = \{11, \ldots, 19\}; \\
& N(c_3) = \{6, \ldots, 15\}; \quad && N(c_4) = \{2, 3, 4, 16, 17, 18, 19\}.
\end{align*}
For this profile and for $k=2$ the Greedy Monroe rule first picks $c_3$, who is approved by 10 voters, then removes these 10 voters, and picks $c_4$. However, committee $\{c_3, c_4\}$ is dominated by $\{c_1, c_2\}$. The same example shows that seq-CC and seq-PAV do not satisfy Pareto efficiency.
\end{example}

\begin{repproposition}{prop:AV_CC_PAV_satisfy_eff}
\propositionParetoSatisfiedAVCCPAVMAV
\end{repproposition}
\begin{proof}
Let $\calR \in \{\text{AV}, \text{CC}, \text{PAV}, \text{$p$-Geometric}\}$.
For the sake of contradiction let us assume that there exists $k \in \naturals$, profile $A \in \calA$, and a committee $W \in \mathcal{S}_k(C)$ such that $W$ dominates each committee from $\calR(A, k)$. In particular, this means that $W$ has strictly lower score than some committee $W_{\opt} \in \calR(A, k)$. Thus, there exists a voter $i \in N$ that assigns to $W_{\opt}$ a higher score than to $W$. However, this is not possible since for each of the considered rules the score that $i$ assigns to a committee $W'$ is an increasing function of $|W' \cap A(i)|$, a contradiction.  

Now, consider MAV, and towards a contradiction assume that there exists $k \in \naturals$, profile $A \in \calA$, and a committee $W \in \mathcal{S}_k(C)$ such that $W$ dominates each committee from $\text{MAV}(A, k)$. Consider a committee $W' \in \text{MAV}(A, k)$. Since $W$ dominates $W'$ for each voter $i \in N$ it holds that $d_\hamming(A(i), W) \leq d_\hamming(A(i), W')$. Thus, $\max_{i \in N} d_\hamming(A(i), W) \leq \max_{i \in N} d_\hamming(A(i), W')$. Consequently, $W$ must be a winning committee according to MAV, and so, $W$ must dominate itself---a contradiction.
\end{proof}

\end{document}